\documentclass[submission,copyright,creativecommons]{eptcs}
\providecommand{\event}{DCM'23} 

\usepackage{iftex}

\ifpdf
  \usepackage{underscore}         
  \usepackage[T1]{fontenc}        
\else
  \usepackage{breakurl}           
\fi

\usepackage{marginnote}
\usepackage{url}
\usepackage{xspace}
\usepackage{amsmath}
\usepackage{amssymb}
\usepackage{amsthm}
\usepackage{wasysym}

\usepackage[noxy]{virginialake}
\usepackage{cleveref}
\usepackage{cmll}
\usepackage{proofzilla}
\usepackage{listings}

\usepackage{adjustbox}

\def\gG{\mathcal G}

\def\N{\mathbb N}

\def\atoms{\mathcal A}
\def\cneg#1{ #1^\perp}

\def\switch{\sigma}

\def\set#1{\{#1\}}
\def\Set#1{\begin{Bmatrix}#1\end{Bmatrix}}
\def\part#1{\left[#1\right]}

\def\block#1{\left(#1\right)}
\def\tuple#1{\left\langle#1\right\rangle}

\def\partset{\mathbb P}
\def\partof#1{\partset_{#1}}
\newcommand{\partofn}[1][n]{\partset_{#1}}

\def\MLL{\mathsf{MLL}}

\def\lpar{\parr}
\def\ltens{\otimes}
\def\axrule{\mathsf {ax}}
\def\cutr{\mathsf {cut}}
\def\mixr{\mathsf {mix}}

\newcommand{\vertices}[1][]{V_{#1}}

\newcommand{\uedge}[1][]{\mkern1mu\mathord{\stackrel{#1}{{\frown}}}\mkern1mu}

\def\con{\mathsf C}

\def\DNF{$\dnf$\xspace}
\def\CNF{$\cnf$\xspace}

\def\cnf{\mathsf{CNF}}
\def\dnf{\mathsf{DNF}}

\newcommand{\intset}[2]{\set{#1, \dots, #2}}

\newcommand{\sbp}[2]{{{\mathfrak B_{\tuple{#1,#2}}}}}
\newcommand{\psbp}[2]{\mathfrak B_{\tuple{#1,#2}}^\perp}
\newcommand{\gsbp}[2]{\mathfrak G_{\tuple{#1,#2}}}
\newcommand{\gsbpp}[2]{\mathfrak G_{\tuple{#1,#2}}^\perp}

\def\dperp{\mathord{\rotatebox[origin=c]{90}{$\mathrel{|}\joinrel \Relbar$}}}

\def\gircon{\mathsf G}

\newcommand{\linkbox}[2]{
	\tikz[overlay,remember picture]
	\draw[thick,rounded corners=2, densely dotted](#1.south west) 	rectangle (#2.north east);
}

\pzvertex{irt}{~}{}

\def\lmtens{\ltens_n}
\def\lmpar{\lpar_n}
\def\lmtensp{{\ltens}^\bullet_n}
\def\lmparp{{\lpar}^\bullet_n}

\def\basicil{{basic intersection link}\xspace}
\def\basicul{{basic union link}\xspace}

\def\basiculs{{basic union links}\xspace}

\def\unfold{expansion\xspace}

\def\hmod{\mathcal H}
\def\rmod{\mathcal R}
\def\cmod{\mathcal C}

\def\inp{\mathsf i}
\def\outp{\mathsf o}

\def\bord{\mathsf B}

\def\link{\ell}

\def\defn#1{\textbf{#1}}

\def\bordof#1{\bord_{#1}}

\def\inpof#1{{\mathsf I_{#1}}}
\def\outpof#1{{\mathsf O_{#1}}}

\def\swof#1{{\behof{#1}}}

\def\prs{\mathcal S}
\def\prn{\mathcal P}

\def\switchset#1{\Sigma(#1)}

\def\mod{\mathcal M}
\def\tmod{\mathcal T}

\def\behof#1{\behsym_{#1}}
\def\behsym{\mathfrak B}
\usepackage{graphicx} 

\def\comp#1{\lseq_{#1}}
\def\parallel{\|}

\def\mllmod{{module}\xspace}

\def\mllps{{proof structure}\xspace}
\def\mllpss{{proof structures}\xspace}
\def\mllpn{{proof net}\xspace}
\def\mllpns{{proof nets}\xspace}

\def\mcomp{component\xspace}
\def\mcomps{components\xspace}
\def\mmod{module\xspace}
\def\mmods{modules\xspace}
\def\mstr{{multiplicative structure}\xspace}
\def\mstrs{{multiplicative structures}\xspace}
\def\mnet{{multiplicative net}\xspace}

\def\ms{\mstr}
\def\mss{\mstrs}
\def\mn{\mnet}

\def\transitory{\mathsf T}
\def\tracomp{$\transitory$-component\xspace}

\newcommand{\restr}[2]{{#1}_{|_{#2}}}
\newcommand{\restrp}[2]{({#1})_{|_{#2}}}
\def\wperp{\perp_w}

\def\sizeof#1{{|{#1}|}}
\def\widecneg#1{(#1)^\perp}

\def\ip{\mathsf{i}}
\def\op{\mathsf{o}}

\def\cyofn{\mathsf{Cy}_n}

\def\gH{\mathcal{H}}

\newcommand{\edge}[1][]{E_{#1}}
\newcommand{\hedge}[1][]{E_{#1}}

\def\hed{\mathsf h}
\def\iof#1{\mathsf{in}(#1)}
\def\oof#1{\mathsf{out}(#1)}
\def\bof#1{\bord(#1)}

\pzemptyvertex{for}{}
\pzemptyvertex{block}{}

\def\MLLx{\MLL^\circ}

\def\myparagraph#1{\noindent\textbf{#1.}}

\newgate{mtens}{\lmtens}{}
\newgate{mpar}{\lmpar}{}

\newgate{tensp}{\lmtensp}{}
\newgate{parp}{\lmparp}{}

\newgate{parpone}{\lpar_1^\bullet}{}
\newgate{parptwo}{\lpar_2^\bullet}{}
\newgate{parpthree}{\lpar_3^\bullet}{}
\newgate{tensptwo}{\ltens_2^\bullet}{}

\def\quand{\quad\mbox{and}\quad}

\def\gM{\mathcal M}
\def\gS{\mathcal S}

\theoremstyle{plain}
\newtheorem{theorem}{Theorem}
\newtheorem{proposition}[theorem]{Proposition}

\newtheorem{corollary}[theorem]{Corollary}

\theoremstyle{definition}
\newtheorem{definition}[theorem]{Definition}
\newtheorem{example}[theorem]{Example}
\newtheorem{nota}[theorem]{Notation}
\newtheorem{remark}[theorem]{Remark}

\pzemptyvertex{Gmod}{execute at begin node=\notsotiny,draw,minimum size=\gateD,shape border uses incircle,inner sep=-4pt,	text centered,shape = isosceles triangle, isosceles triangle apex angle = 90,gatecorners,shape border rotate=-90
}
\pzemptyvertex{Tmod}{
	execute at begin node=\notsotiny,draw,minimum size=\gateD,shape border uses incircle,inner sep=-4pt,text centered,basicpsgate,shape = trapezium, trapezium angle=110, gatecorners
}

\def\prolog#1{\mathbf{#1}}
\def\colonminus{:-~}
\def\gluing#1#2#3{#1 \lseq_{#3} #2}
\def\vG{\vertices[\gG]}
\def\vH{\vertices[\gH]}

\title{Logic Programming with Multiplicative Structures}
\author{Matteo Acclavio\footnote{The author is supported by Villum Fonden, grant no. 50079.}
\institute{University of Sussex, Brighton (UK)\\
and\\University of Southern Denmark, Odense (DK)}
\and
Roberto Maieli\footnote{The author is supported by INdAM-GNSAGA.}
\institute{Dipartimento di Matematica e Fisica
	 \\ Università Roma Tre \\ Roma, Italy}
}
\def\titlerunning{Logic Programming with Multiplicative Structures}
\def\authorrunning{M. Acclavio \& R. Maieli}

\begin{document}
\maketitle

\begin{abstract}
    In the logic programming paradigm, a program is defined by a set of methods, each of which can be executed when specific conditions are met during the current state of an execution. The semantics of these programs can be elegantly represented using sequent calculi, in which each method is linked to an inference rule. In this context, proof search mirrors the program's execution.
	Previous works introduced a framework in which the process of constructing proof nets is employed to model executions, as opposed to the traditional approach of proof search in sequent calculus.
	
	This paper further extends this investigation by focussing on the pure multiplicative fragment of this framework. 
	We demonstrate, providing practical examples, 
	the capability to define {\em logic programming methods} with context-sensitive behaviors solely through specific resource-preserving and context-free operations, corresponding to certain generalized multiplicative connectives explored in existing literature. We show how some of these methods, although still multiplicative, escape the purely multiplicative fragment of Linear Logic (MLL, containing  only  $\lpar$ and $\ltens$).
\end{abstract}

\section{Introduction}

Proof theory provides various paradigms for interpreting computations as proofs and their transformations. The renowned Curry-Howard correspondence interprets proofs as programs, and proof reduction (i.e., cut-elimination) as program execution. 
This correspondence offers an elegant model for \emph{functional programming}, where the primary computation mechanism is substitution.
In this context, well-typed programs are expected to terminate their execution after computing results derived by the complete initial information.
However, this paradigm appears to face challenges in representing programs where the main computational mechanism is not substitution, as well as the ones characterized by non-termination, partial information, and strong concurrency (see, e.g., distributed systems, database servers, and microservices architectures).
By means of example, consider the way of modeling the Curry-Howard correspondence in the case of non-terminating programs, where 
we need to consider \emph{infinitary proof systems} to be able to represent infinite programs as \emph{non-wellfounded derivations}.
In these systems, even basic results like soundness, completeness, and cut-elimination require complex techniques \cite{bae:dou:sau:infinitary,Das2021,das:pous:non-well,acc:cur:gue:infext,sau:23,acc:cur:gue:infinitary}.
Therefore, it may appear more intuitive to interpret the rules of operational semantics for these programs as rules of a sequent system, and the program execution as the process of proof search \cite{miller:overview,miller1995survey,and:01,and:maz:03}.
This approach naturally handles issues related to partial information, the concurrency of rule applications, and the possibility of non-termination.

Alternatively, in the \emph{logic programming} paradigm, a program is provided by a set of \emph{methods}, which are elementary programs that can be executed when specific preconditions are met. 
The conventional proof-theoretical interpretation of logic programming associates a sequent of formulas with each state of the computation, and a sequent calculus rule with each program method.
This establishes an intuitive connection between program execution and the process of proof search within the calculus. 
In this context, a derivation tree where all leaves are axioms of the system represents a successfully completed computation.

It's worth noting that two forms of non-determinism arise during the process of proof search, corresponding to two distinct notions of non-determinism in program executions. 
Using the terminology from \cite{andreoli1992logic,hemer2002don,liang:hal-03457379},
The first type of non-determinism arises from the possibility of applying multiple methods to separate sub-sequents, which is a consequence of the limitations of sequent calculus%
\footnote{
    In sequent calculus, two independent rules which can be applied to a same sequent must be sequentialized because the syntax does not allow for the application of rules to portions of a sequent.
    At the same time, if proof search produces a branching, then the two branches of the proof search can be performed independently in a true concurrent way.
}%
. 
The second form of non-determinism is observed when different methods are applied to non-disjoint sequents.

In this paper we continue the investigation on the interpretation of logic programming based on linear logic \emph{\mllps expansion} instead of sequent calculus proof search, as in \cite{and:02,and:maz:03,fou:mog04,hae:fag:sol:07}.
In this approach, 
the set of inputs of a \mllps is interpreted as the current state of an execution, 
and 
the process of connecting new proof structures to its input (called \emph{expansion}) is interpreted as the application of a method.
The motivation to employ proof structures is due to their efficacy in capturing the non-determinism arising from the constraints of sequent calculus syntax:
the graphical syntax relieves us from the bureaucracy of rules permutations between independent rule applications. 
Additionally, proof structures offer a more flexible structure, enabling us to define methods corresponding to the expansion of multiple branches simultaneously.

\begin{figure}[t]
	\centering
	\noindent\resizebox{.9\textwidth}{!}{
		\begin{tabular}{|c|l|l|l|}
			\hline
			Paradigm
			&
			Curry-Howard
			&
			Logic programming
			&
			Logic programming
			\\
			&
			&
			(sequent calculus)
			&
			(multiplicative structures)
			\\
			\hline
			State
			&
			Proof
			&
			Proof
			&
			Multiplicative structure:
			\\
			&
			- with cuts;
			&
			- without cuts;
			&
			- (transitory) component;
			\\
			&
			- without proper axioms
			&
			- possibly with proper axioms
			&
			- possibly with inputs;
			\\
			&
			(i.e., without open premises);
			&
			(i.e., with open premises);
			&
			\\
			\hline
			Computation
			step
			&
			Proof reduction:
			&
			Proof construction:
			&
			Proof net expansion:
			\\
			& 
			Cut-elimination
			&
			Proper axioms elimination
			&
			Proof structure expansion
			\\
			\hline
			Final state
			&
			Cut-free proof
			&
			Derivation with closed branches
			&
			Structure without inputs
			\\
			\hline
			Program type
			&
			Formula
			&
			Formula
			&
			Network behavior
			\\
			\hline
		\end{tabular}    
	}
	\caption{
		A summary of the interpretation of proofs-as-programs in the paradigms of functional programming, and logic programming using sequent calculus and using proof net expansion.
	}
	\label{fig:paradigms}
\end{figure}

\myparagraph{Contributions of the paper}
%
In this paper, we study a logic programming framework built upon linear logic proof structures, offering a generalization of the standard \emph{$\MLL$-proof structure} \cite{gir:ll} and the {focused bipolar proof structures} \cite{and:maz:03}. 
We focus on the multiplicative fragment of this framework, where the Danos-Regnier \emph{correctness criterion} \cite{dan:reg:89} can be easily generalized.

We introduce the concept of a \emph{\mcomp} as an acyclic multiplicative structure where each of its part can interact with a context, analogous to the notion of an \emph{open derivation} in sequent calculus. 
After establishing the topological conditions necessary for ensuring the composability of components, we proceed to define the foundational blocks of a logic programming framework based in the expansion of proof structures.
Within this framework, as  main novelty, we offer a computational interpretation of a specific family of \emph{generalized multiplicative connectives}, which are connectives provided with linear and context-free introduction rules proposed in early works on linear logic \cite{dan:reg:89,gir:meanII}, 
but which lacked of any concrete computational interpretation prior to this work.
We conclude by illustrating methods, defined within a linear and context-free setting, whose behavior is ``locally additive'' (in the sense of \cite{gir:meanII}), which cannot be expressed using the conventional $\MLL$ connectives $\lpar$ and $\ltens$ (see \cite{mai:14,acc:mai:20, mai:22}).

\myparagraph{Outline of the paper}
%
In \Cref{sec:back} we recall some notions on hypergraphs, partitions and the definition of multiplicative linear logic and its proof nets.
In \Cref{subsec:genCon} we recall the definition and results from \cite{mai:14,acc:mai:20} on the generalized multiplicative connectives (theorized in \cite{dan:reg:89,gir:meanII}) we use in this paper.
In \Cref{sec:logProg} we give an overview of the way logic programming program executions are represented in using sequent calculi and how this paradigm has been extended in \cite{and:maz:03} to proof net expansion.
In \Cref{sec:gen} we define our framework by extending the definition of proof structures, and proving the results about their compositionality.
In \Cref{sec:examples} we provide a computational interpretation of certain generalized multiplicative connectives in our framework.
We show that these connectives, beside being linear and context-free operators, are still able to capture non-linear and context-sensitive behaviors.

\section{Preliminary Notions}\label{sec:back}

In this section we recall basic definitions for hypergraphs and partitions we use in this paper.
We then recall the definition of multiplicative linear logic and the syntax of its proof nets.

\subsection{Hypergraphs}\label{sec:hyp}

	A \defn{hypergraph} $\gG=\tuple{\vertices[\gG],\hedge[\gG]}$ 
	is given by 
	a set of \defn{vertices} $\vertices[\gG]$ and 
	a set of \defn{hyperedges} $\hedge[\gG]$, that is, a set of pairs of list of vertices in $\vertices[\gG]$.
	In a hyperedge $\hed=\tuple{\iof\hed,\oof\hed}\in\hedge[\gG]$ we call the vertices occurring in $\iof\hed$ (resp.~ in $\oof\hed$) the \defn{inputs} (resp.~the \defn{outputs}) of $\hed$
	and
	we define the \defn{border} $\bof\hed$
	as the multiset of vertices occurring in $\iof\hed$ and in $\oof\hed$.
	A \defn{sub-hypergraph} of a hypergraph $\gG=\tuple{\vertices[\gG],\hedge[\gG]}$  is a hypergraph $\gG'\tuple{\vertices[\gG'],\hedge[\gG']}$ such that $\vertices[\gG']\subseteq\vertices[\gG]$ and $\hedge[\gG']\subseteq\set{\hed\in \hedge[\gG]\mid \bof{\hed}\subseteq\vertices[\gG']}$.
	
	An \defn{input} (resp.~\defn{output}) of the hypergraph $\gG$ is a vertex which does not occur as output (resp.~input) of any hyperedge of $\gG$.
	We denote by $\inpof\gG$ (resp.~$\outpof\gG$) the set of inputs (resp.~outputs) of $\gG$, and we define the \defn{border} of $\gG$ as the set of its inputs and the outputs, that is, $\bordof{\gG}=\inpof\gG\cup\outpof\gG$.
	A hypergraph is \defn{linear} if each vertex occurs at most once as an input of a hyperlink and as an output of another link%
	\footnote{
		In works on hypergraphs with interfaces (e.g., \cite{bonchi2016}), this property is referred to as \emph{linearity} or \emph{monogamicity}.
		Note that, by definition, no vertex can occur at the same time as an input and an output of a linear hyperedge.
	}%
	.
	
	In a hypergraph $\gG$, a \defn{path} (of length $n$) from $x\in \vertices[\gG]$ to $y\in\vertices[\gG]$ 
	is an alternating list of vertices and hyperedges
	of the form
	$x=v_0 \hed_1 v_1 \cdots \hed_n v_n=y$
	such that 
	$v_{i-1}=\oof{\hed_i}$ and $v_i=\iof{\hed_i}$ 
	for all $i\in \intset1n$; in this case we say that $x$ is connected to $y$.
	A \defn{cycle} is a path with $\hed_1=\hed_n$, or $n>0$ and $v_0=v_n$; 
	it is \defn{elementary} if there are no $i$ and $j$ such that $i\neq j$ and $v_i=v_j$ or $\hed_i\neq \hed_j $.
	A hypergraph is \defn{acyclic} if it contains no elementary cycles.

	An \defn{undirected hypergraph} 
	$\gG=\tuple{\vertices[\gG],\hedge[\gG]}$
	is given by 
	a set of \defn{vertices} $\vertices[\gG]$ and 
	a set of \defn{undirected hyperedges} $\hedge[\gG]$, that is, a set of subset of vertices in $\vertices[\gG]$.
	A \defn{graph} is an undirected hypergraph in which each hyperedge is an \defn{edge}, that is, a set of two vertices $\set{v_1,v_2}$.
	\defn{Paths} and \defn{cycles} in an undirected hypergraph are defined analogously to the ones in hypergraph.
	Two vertices are \defn{connected} if there is a path from one to the other.
	A \defn{connected component} of an undirected hypergraph is a maximal subset of pairwise connected vertices.
	The \defn{undirected hypergraph} associated to a linear hypergraph $\gG=\tuple{\vertices[\gG],\hedge[\gG]}$ is defined as the undirected hypergraph
	$\tuple{\vertices[\gG],\set{\bof\hed\mid\hed\in\hedge[\gG]}}$.
	We say that a hypergraph $\gG$ is \defn{connected} if the undirected hypergraph associated to it is connected.

\begin{nota}
	We drawing hyperedges with inputs on top and output on the bottom. 
	By convention, we enumerate the inputs and outputs from left to right.
\end{nota}

\begin{definition}\label{def:seqcomp}
    Let $\gG=\tuple{\vertices[\gG],\hedge[\gG]}$ and $\gH=\tuple{\vertices[\gH],\hedge[\gH]}$ be two hypergraphs with disjoint sets of vertices.
	The \defn{disjoint union} of $\gG$ and $\gH$ (denoted $\gG\parallel\gH $) is defined as the union of vertices and hyperedges of $\gG$ and $\gH$
	that is, 
	$\gG\parallel\gH 
	\coloneqq 
	\tuple{\vG\cup\vH,\hedge[\gG]\cup\hedge[\gH]}$.
	Note that in defining $\gG\parallel \gH$ we always assume $\gG$ and $\gH$ having disjoint sets of nodes.
	An \defn{(linear) interface} $X=(g,h)$ is a pair of bijections from a finite set $\intset1n$ to $\vG$ and $\vH$ respectively.
	We define the \defn{composition} of $\gG$ and $\gH$ via an interface $X$ 
	as the hypergraph $\gluing\gG\gH X$ 
	obtained by identifying in $\gG\parallel\gH$ the vertices $g(i)$ and $h(i)$ for each $i\in\intset1n$, as shown on the left of \Cref{eq:interface}.
	\gdef\vwi#1{\mathord
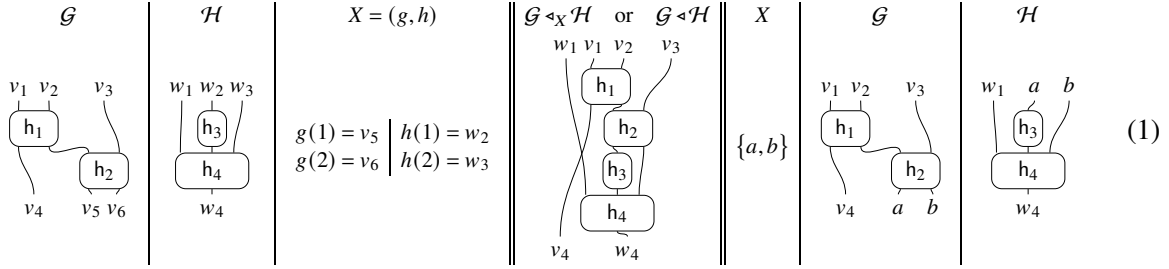
{%
			\tikz[remember picture,baseline=(w#1.base)]%
			\node[inner sep=0pt](w#1){$w_#1\strut$};
	}}
	\gdef\vvi#1{\mathord{%
		\tikz[remember picture,baseline=(v#1.base)]%
		\node[inner sep=0pt](v#1){$v_#1\strut$};
	}}
	\begin{equation}\label{eq:interface}
	\adjustbox{max width=.92\textwidth}{$
		\begin{array}{c|c|c||c||c|c|c}
			\gG & \gH & X=(g,h) &
			\gluing\gG\gH X \quad\mbox{or}\quad \gG\lseq\gH 
			& X &\gG & \gH
			\\
			\begin{array}{cc}
				\vvi1 \;\; \vvi2 & \vvi3
				\\[5pt]
				\vmod1{\;\;\hed_1\;\;}
				\\[5pt]
				&\vmod2{\;\;\hed_2\;\;}
				\\[5pt]
				\vvi4&\vvi5 \;\vvi6
			\end{array}
			\pswires{v1/mod1.130,v2/mod1.50,mod1.310/mod2.130,v3/mod2.50,mod1.230/v4,mod2.230/v5,mod2.310/v6}
			&
			\begin{array}{c}
				\vwi1 \; \vwi2 \; \vwi3
				\\[5pt]
				\vmod3{\hed_3}
				\\[5pt]
				\vmod4{\quad\hed_4\quad}
				\\[5pt]
				\vwi4
			\end{array}
			\pswires{w1/mod4.150,w2/mod3,w3/mod4.40,mod3/mod4,mod4/w4}
			&
			\begin{array}{c|c}
				g(1)=v_5
				&
				h(1)=w_2
				\\
				g(2)=v_6
				&
				h(2)=w_3
			\end{array}
			&
			\begin{array}{c}
				\vwi1\; \vvi1 \;\; \vvi2 \quad \;\vvi3
				\\[5pt]
				\vmod1{\;\;\hed_1\;\;}\quad
				\\[5pt]
				\quad\vmod2{\;\;\hed_2\;\;}
				\\[5pt]
				\vmod3{\hed_3}
				\\[5pt]
				\vmod4{\quad\hed_4\quad}
				\\[5pt]
				\vvi4 \qquad \vwi4\qquad
			\end{array}
			\pswires{v1/mod1.130,v2/mod1.50,mod1.310/mod2.130,v3/mod2.50,mod1.230/v4}
			\pswires{w1/mod4.150,mod3/mod4,mod4/w4}
			\pswires{mod2.230/mod3,mod2.310/mod4.40}
			&
			\Set{a,b}
			&
			\begin{array}{cc}
				\vvi1 \;\; \vvi2 & \vvi3
				\\[5pt]
				\vmod1{\;\;\hed_1\;\;}
				\\[5pt]
				&\vmod2{\;\;\hed_2\;\;}
				\\[5pt]
				\vvi4&\va5 \quad \vb6
			\end{array}
			\pswires{v1/mod1.130,v2/mod1.50,mod1.310/mod2.130,v3/mod2.50,mod1.230/v4,mod2.230/a5,mod2.310/b6}
			&
			\begin{array}{c}
				\vwi1 \quad \va2 \quad \vb3
				\\[5pt]
				\vmod3{\hed_3}
				\\[5pt]
				\vmod4{\quad\hed_4\quad}
				\\[5pt]
				\vwi4
			\end{array}
			\pswires{w1/mod4.150,a2/mod3,b3/mod4.40,mod3/mod4,mod4/w4}
		\end{array}
	$}
	\end{equation}
	To improve readability, 
	we may simply identify the vertices $g(i)$ and $h(i)$ for all $i\in\intset1n$, therefore considering $X$ as a set of vertices in $\vertices[\gG]\cap\vertices[\gH]$
	and simply writing
	$\gG\lseq \gH$, as shown in the right of \Cref{eq:interface}.
\end{definition}

\subsection{Partitions}\label{subsec:back:part}

Given a set $X$, 
a \defn{partition} of $X$ is a set of disjoint subsets of $X$ (we call each a \defn{block}) such that their union is $X$.
In order to improve readability, when writing sets of partitions, in which three parentheses are nested inside each other, 
even if blocks and partitions are sets (not permutations, nor multisets),
we use parentheses $\block{-}$ to denote blocks (subsets of $X$), and square brackets $\part{-}$ to denote partitions (sets of subsets of $X$).
For example, we write $\part{\block{1,3},\block 2}$ to denote 
the partition of the set $\set{1,2,3}$ with one block containing $1$ and $3$ and one block containing only $2$.
We denote by $\partof X$ (resp.~$\partofn$) the set of partitions over $X$ (resp.~over $\intset{1}{n}$).
If $p\in \partof X$ and $Y\subset X$, we define the \defn{restriction of $p$ on $Y$} as the partition $\restr p Y\in \partof Y$ such that $x,y\in Y$ belongs to the same block $\restr \gamma Y \in \restr p Y$ iff $x$ and $y$ belongs in a same block $\gamma\in p$.
By means of example, if $p= \part{\block{1,3,4},\block{2,5}, \block 6}\in \partofn[6]$, 
then $\restr p {\set{1,2,3}}=\part{\block{1,3},\block 2}$.

\begin{definition}[Orthogonality \cite{dan:reg:89}]\label{def:ort}
	Let $X$ be a set and $p,q\in\partof X$.
	We define \defn{graph of incidence of $p$ and $q$} 
	as the graph $\gG(p,q)$ with 
	vertices the blocks in $p$ and in $q$
	and with
	an edge between a block $\gamma_p\in p$ and a block in $\gamma_q\in q$ for each $i\in \gamma_p\cap \gamma_q$ (see examples in \Cref{eq:orth}).
	That is, the graph $\gG(p,q)=\tuple{\vertices[\gG(p,q)], \edge[\gG(p,q)]}$ has set of vertices and edges respectively
	$$
	\vertices[\gG(p,q)]=\set{\gamma\mid \gamma \in p\mbox{ or } \gamma \in q}
	\quand
	\edge[\gG(p,q)]=\set{\set{\gamma^p_i,\gamma^q_i} \mid \gamma^p_i\in p \mbox{ and }  \gamma^q_i\in q \mbox{ and } i \in \gamma^p_i \cap \gamma^q_i}
	\;.
	$$
	We say that $p,q\in \partof X$ are \defn{weakly orthogonal}, denoted $p\wperp q$,
	if their graph of incidence  $\gG(p,q)$ is acyclic.
	They are \defn{orthogonal}, denoted $p\perp q$, 
	if their graph of incidence is connected and acyclic.
	
	The notion of weak orthogonality and orthogonality extends to sets of partitions:
	if $P,Q\subset \partof X$ then $P\wperp Q$ 
	(respectively $P\perp Q$) 
	if $p\wperp q$ (respectively $p\perp q$) for all $p\in P$ and for all $q\in Q$.
	The \defn{orthogonal set} of a set of partitions  $P\subset \partofn$ is defined as $P^\perp=\set{q\in \partofn \mid p \perp q\mbox{ for all }p\in P }$.
	We write $P \dperp Q$ if $P\perp Q$ and $P^\perp\perp Q^\perp$.
\end{definition}

\begin{example}
	Consider 
	the partitions 
	$p=\part{\block{1,2},\block 3}$, $q_1=\part{\block{1,2,3}}$,
	$q_2=\part{\block{1,3},\block 2}$
	$q_3=\part{\block{1},\block{2,3}}$ and $q_4=\part{\block 1, \block 2, \block 3}$.
	We have that
	$p\not\perp q_1$,
	$p\perp q_2$,
	$p\perp q_3$,
	$p\wperp q_4$
	because

	\begin{equation}\label{eq:orth}
	\adjustbox{max width=.9\textwidth}{$
	\begin{array}{c@{\quad}|@{\quad}c@{\quad}|@{\quad}c@{\quad}|@{\quad}c}
		\begin{array}{ccc}
			\vblock1{\block{1,2}}&&\vblock2{\block{3}}
			\\[1em]
			&\vblock3{\block{1,2,3}}
		\end{array}
		\Gedges{block2/block3}
		\bentGedges{block1/block3/20,block1/block3/-20}
		& 
		\begin{array}{c@{\qquad}c}
			\vblock1{\block{1,2}}&\vblock2{\block{3}}
			\\[1em]
			\vblock3{\block{1,3}}&\vblock4{\block{2}}
		\end{array}
		\Gedges{block1/block3,block1/block4,block2/block3}
		& 
		\begin{array}{c@{\qquad}c}
			\vblock1{\block{1,2}}&\vblock2{\block{3}}
			\\[1em]
			\vblock3{\block{1}}&\vblock4{\block{2,3}}
		\end{array}
		\Gedges{block1/block3,block1/block4,block2/block4}
		& 
		\begin{array}{ccc}
			\vblock1{\block{1,2}}&&\vblock2{\block{3}}
			\\[1em]
			\vblock3{\block{1}}&\vblock4{\block{2}}&\vblock5{\block{3}}
		\end{array}
		\Gedges{block1/block3,block1/block4,block2/block5}
		\\
		\mbox{is cyclic}
		&
		\mbox{is connected and acyclic}
		&
		\mbox{is connected and acyclic}
		&
		\mbox{is acyclic.}
	\end{array}
	$}
\end{equation}
\end{example}

\subsection{Multiplicative Linear Logic}\label{sec:MLL}

Multiplicative linear logic formulas are generated from a countable set $\atoms=\set{a, b, \dots}$ of \defn{propositional variables}  by the following grammar:
$$A,B::= a \mid \cneg A \mid A\lpar B \mid A\ltens B$$
We consider formulas modulo the \defn{involutivity of the negation} $A^{\perp\perp}=A$ and
the \defn{De Morgan} laws 
$\widecneg{A\ltens B}=\cneg A \lpar \cneg B$ and $\widecneg{A\lpar B}=\cneg A \ltens \cneg B$.
%

A \defn{sequent} is a set of occurrences of formulas 
(as in, e.g., \cite{bae:dou:sau:infinitary,acc:cur:gue:infinitary}).
The \defn{sequent systems} $\MLL=\set{\axrule, \lpar, \ltens}$ and $\MLLx=\MLL\cup\set{\mixr}$ are defined using the rules in \Cref{fig:mllrules}%
\footnote{
	In the figure we include the rule $\cutr$ required to define compositionality of proofs via \emph{modus ponens}.
	We do not include it in the definition of $\MLL$ and $\MLLx$
	since 
	this rule is proven to be admissible \cite{gir:ll}
	and 
	it plays no role in the framework we are presenting in this paper.
}. 
We call \defn{active} (resp.~\defn{principal}) a formula occurrence in one of the premises (resp.~in the conclusion) of a rule, not occurring in conclusion (resp.~in any of its premises). 

\begin{figure}[t!]
	\centering
	\begin{tabular}{c@{\qquad}c@{\qquad}c@{\qquad}|@{\qquad}c@{\qquad}|@{\qquad}c}
		$\vlinf{}{\axrule}{\vdash a, \cneg a}{}$
		&
		$\vlinf{}{\lpar}{ \vdash\Gamma, A \lpar B}{ \vdash\Gamma, A , B}$
		&
		$\vliinf{}{\ltens}{ \vdash \Gamma, A \ltens B, \Delta}{\vdash\Gamma, A}{ \vdash B , \Delta}$
		&
		$\vliinf{}{\mixr}{ \vdash\Gamma, \Delta}{\vdash \Gamma}{\vdash\Delta}$
		&
		$\vliinf{}{\cutr}{ \vdash \Gamma, \Delta}{ \vdash\Gamma, A}{ \vdash\cneg A , \Delta}$
	\end{tabular}
	\caption{Sequent calculus rules for the multiplicative linear logic,
		and rules $\mixr$ and $\cutr$.}
	\label{fig:mllrules}
\end{figure}

\paragraph*{Multiplicative Proof Nets}

\emph{Proof nets} are a graphical syntax for multiplicative linear logic proofs capturing the proof equivalence generated by independent rules permutations (see, e.g., \cite{gir87,hug:inv,hei:hou:noPN,hei:hug:MALL,acc:str:fromto,ICP,abr-mai-2019,acc:str:AIML22}).
%
%
    A \defn{proof structure} $\prs=\tuple{\vertices[\prs],\hedge[\prs]}$ is a  hypergraph 
    whose vertices are labeled by $\MLL$-formulas and whose hyperedges (called links) are labeled by rules in $\MLL$ (such labels are called \defn{types}) in such a way the following local constraints are respected:
	\begin{equation}\label{eq:MLLlinks}
		\adjustbox{max width=.9\textwidth}{$\begin{array}{c@{\qquad}c@{\qquad}c@{\qquad}c}
			\begin{array}{c@{\qquad}c@{\qquad}c}
				\\
				\va1 &&\vna1
			\end{array}
			\psaxioms{a1/na1}
			&
			\begin{array}{c@{\qquad}c@{\qquad}c}
				\vA1 &&\vB1\\[1em]
				&\Gtens1\\
				& \vfor1{A\ltens B}
			\end{array}
			\pswires{A1/Gtens1.L,B1/Gtens1.R,Gtens1.O/for1}
			&
			\begin{array}{c@{\qquad}c@{\qquad}c}
				\vA1 &&\vB1\\[1em]
				&\Gpar1\\
				& \vfor1{A\lpar B}
			\end{array}
			\pswires{A1/Gpar1.L,B1/Gpar1.R,Gpar1.O/for1}
		\end{array}$}
	\end{equation}
with $A$ and $B$ formulas and $a\in\atoms$.

Since we are considering a sequent as a set of occurrences of
formulas, it is possible to easily trace formula occurrences in a
derivation, defining a proof structure that encodes a given
derivation.

\begin{definition}
	Let $\pi$ be a derivation in $\MLL$.
	We define the \mllps \defn{representing} $\pi$ as the \mllps $\prn_\pi$ having
	a vertex for each occurrence of an active formulas of a rule in $\pi$,
	and
	a link of type $\rho$ with inputs (resp.~with output) the vertices corresponding to the active formulas (resp.~the principal formula) for each occurrence of a rule $\rho$ in $\pi$.
	A \defn{\mllpn} is a \mllps $\prs=\prn_\pi$ representing a derivation $\pi$ in $\MLL$.
\end{definition}

By definition not all \mllpss are \mllpns.
For this reason, 
a \defn{correctness criterion}, that is, 
a topological characterization of 
those \mllpss which are \mllpns,
is needed.
Beside various criteria have been developed in the literature \cite{gir:ll,dan:reg:89,retore:phd,gue:99}, 
we here report the so-called \defn{Danos-Regnier criterion} (or \defn{DR-criterion} for short), which is the most relevant to our purposes.

\def\DRcor{DR-correct\xspace}
\begin{definition}\label{def:netCor}
	Let $\gM$ be a \mllps.
	A \defn{test} for $\gM$ is 
	the undirected hypergraph with the same vertices of $\gM$
	having 
	a hyperedge $\set{a,\cneg a}$ for each $\axrule$-link $(\emptyset,\tuple{a,\cneg a})$,
	a hyperedge $\set{A,B,A\ltens B}$ for each $\ltens$-link $(\tuple{A,B},\tuple{A\ltens B})$, and 
	either one edge $\set{A,A\lpar B}$ or one edge  $\set{B,A\lpar B}$ for each $\lpar$-link $(\tuple{A,B},\tuple{A\lpar B})$.
	The \mllps $\gM$ is \defn{\DRcor} if it has no inputs and if all of its tests are connected and acyclic (undirected) hypergraph.
\end{definition}

\begin{theorem}[\cite{dan:reg:89}]\label{thm:DRcorr}
	A \mllps $\prs$ is a \mllpn iff $\prs$ is \DRcor.
\end{theorem}

It is worth noticing that by dropping connectedness condition in \Cref{def:netCor}, we obtain a notion of correctness for $\MLLx$-\mllpn, 
that is, 
if any test of a \mllps $\gM$ is an acyclic hypergraph, 
then we can associate to 
$\gM$ a derivation in $\MLLx$.

\begin{definition}
	A \defn{module} is a \mllps which is a connected sub-hypergraph of a \mllpn.
\end{definition}

\begin{remark}[Definitions of module in the literature]\label{rem:mod}
	\def\drmodule{$\mathsf {DR}$-module\xspace}
	\def\drmodules{$\mathsf {DR}$-modules\xspace}
	\def\ply{\mathsf Y}
	
	The definition of module we consider in this paper differs from the definition of module given 
	in \cite{dan:reg:89} 
	where a module 
	is defined as a pair $\tuple{\prs, \ply}$
	with $\prs$ a  proof structure 
	such that $\switch(\prs)$ is acyclic for all $\switch\in \switchset{\prs}$, 
	and 
	a subset of its border $\ply\subseteq\bordof\prs$.
\end{remark}
\begin{figure}[t]
\centering
\adjustbox{max width=\textwidth}{
	\begin{tabular}{|c|c|c|}
		\hline
			$\begin{array}{ccc@{\qquad}cccc@{\qquad}cccc}
				\\[3em]
				&\Gpar1&&&&\Gtens1&&&&\Gtens2\\
				&\,\, \vA1_1&&&&\,\,\vA2_2&&&&\,\,\vA3_3
				\\[.5em]
			\end{array}
			\psaxioms{Gpar1.R/Gtens1.L,Gtens1.R/Gtens2.L}
			\psaxiom{Gpar1.L}{Gtens2.R}{.75}{}
			\pswires{Gpar1.O/A1,Gtens1.O/A2,Gtens2.O/A3}
			$
		& 
			$
			\begin{array}{ccc}
				\va1_1 \quad \va2_2&&\va3_3 \quad \va4_4
				\\[1em]
				\Gpar1 & & \Gtens1
				\\
				&\Gtens2
				\\[.5em]
				&\vA1
			\end{array}
			\pswires{a1/Gpar1.L,a2/Gpar1/R,a3/Gtens1.L,a4/Gtens1.R}
			\pswires{Gpar1.O/Gtens2.L,Gtens1.O/Gtens2.R}
			\pswires{Gtens2.O/A1}
        $
		&
        $\begin{array}{ccc@{\qquad}cccc@{\qquad}cccc}
				\\[1em]
				\va1_1	&& \va2_2 &&  \va3_3 &&\va4_4\\[1em]
				&\Gpar1&&&&\Gpar2\\
				&\,\, \vA1_1&&&&\,\,\vA2_2
				\\[.5em]
			\end{array}
			\psaxioms{a2/a3}
			\pswires{a2/Gpar1.R,a3/Gpar2.L}
			\pswires{a1/Gpar1.L,a4/Gpar2.R}
			\pswires{Gpar1.O/A1,Gpar2.O/A2}
			$
		\\
		\hline
		$\prs_1$
		&
		$\prs_2$
		&
		$\prs_3$
		\\
		\hline
	\end{tabular}
}
\caption{Examples of proof structures: $\prs_1$ is a proof net (therefore a module), $\prs_2$ is a module, and $\prs_3$ is not a module (it admits a test in which $a_2$ and $a_3$ are disconnected from any other vertex, therefore $\prs_3$ cannot be a sub-hypergraph of a proof net).}
\label{fig:modEx}
\end{figure}

\section{Generalized Connectives in Multiplicative Linear Logic}\label{subsec:genCon}

The notion of generalized (multiplicative) connectives for multiplicative linear logic was introduced since the early works on linear logic \cite{dan:reg:89}.
We say that an inference rule of the sequent calculus is \defn{linear} if each occurrence of subformula (except the principal formula of the rule) occurring in the conclusion of the rule occurs exactly once in its premise(s), and it is \defn{context-free} if no conditions on the non-principal formulas affect the application of the rule.
A rule is \defn{multiplicative} if linear and context-free.
\begin{example}
	Consider the three rules in \Cref{eq:rulesEx} below.
	Only the leftmost is multiplicative:
	the central one is not linear since the subformula $B$ does not occur in the premise,
	while
	the rightmost one is not context-free since the rule requires the sequent to contain an odd number of formulas.
	\begin{equation}\label{eq:rulesEx}
		\begin{array}{c@{\qquad}c@{\qquad}c}
			\vliinf{}{\mbox{\scriptsize $a \ltens (b\lpar c)$ }}{\vdash\Gamma, \Delta, A \ltens (B\lpar C)}{\vdash\Gamma, A}{\vdash\Delta, B, C}
			&
			\vlinf{}{\mathsf W_\lpar}{\vdash\Gamma, A \lpar B}{\Gamma, A}
			&
			\vliinf{}{\mathsf{odd}_\ltens}{\vdash C_1, C_2, \ldots, C_{2k-1}, C_{2k}, A \ltens B}{\vdash C_1,\dots, C_{2k-1}, A}{\vdash C_2, \dots , C_{2k}, B}
		\end{array}
	\end{equation}
\end{example}

In \cite{dan:reg:89} the authors observe that any multiplicative rule can be fully described by a partition (having a block for each of the rule premises) keeping track of how active formulas are distributed among the premise of the rule.
Thus, we can define so called \defn{synthetic rules} 
(see, e.g., \cite{dan:reg:89,gir:meanII,dale-elaine}),
allowing us to gather multiple inference of multiplicative rules to construct a formula by a single rule application, as shown in the following example.

\begin{example}\label{ex:synthetic}
	Consider the following formulas, their \emph{synthetic rules} and the associated partitions.
	
	\begin{equation*}\adjustbox{max width=\textwidth}{$
		\begin{array}{l|c@{\quad}|@{\quad}c@{\quad}|@{\quad}c}
			\mbox{Formula}
			&
			F=a \ltens (b\lpar c)
			&
			G=(a\lpar b) \ltens (c\lpar d)
			&
			H=a \lpar (b \ltens c)
			\\\hline
			\mbox{Derivation(s)}
			&
			\vlderivation{
				\vliin{}{\ltens}{
					\Gamma,\Delta,a\ltens(b\lpar c)
				}{
					\vlhy{\Gamma,a}
				}{
					\vlin{}{\lpar }{\Delta, b\lpar c}{\vlhy{\Delta, b,c}}
				}
			}
			&
			\vlderivation{
				\vliin{}{\ltens}{
					\Gamma,\Delta,(a\lpar b)\ltens(c\lpar d)
				}{
					\vlin{}{\lpar }{\Delta, a\lpar b}{\vlhy{\Delta, a,b}}
				}{
					\vlin{}{\lpar }{\Delta, c\lpar d}{\vlhy{\Delta, c,d}}
				}
			}
			&
			\vlderivation{
				\vlin{}{\lpar}{\Gamma,\Delta,a\lpar(b\ltens c)}{
					\vliin{}{\ltens}{\Gamma,\Delta, a, b\ltens c}{\vlhy{\Gamma, a, b}}{\vlhy{\Delta, c}}
				}
			}
			\mbox{\quad and \quad}
			\vlderivation{
				\vlin{}{\lpar}{\Gamma,\Delta,a\lpar(b\ltens c)}{
					\vliin{}{\ltens}{\Gamma,\Delta, a, b\ltens c}{\vlhy{\Gamma, a,c}}{\vlhy{\Delta, b}}
				}
			}
			\\\hline
			\mbox{Synthetic Rule(s)}
			&
			\vliinf{}{\mathcal{R}_F}{\vdash\Gamma, \Delta, a \ltens (b\lpar c)}{\vdash\Gamma, a}{\vdash\Delta, b, c}
			&
			\vliinf{}{\mathcal{R}_G}{\vdash\Gamma, \Delta, (a \lpar b) \ltens (c\lpar D)}{\vdash\Gamma, a,b}{\vdash\Delta, c,D}
			&
			\vliinf{}{\mathcal{R}_{H_1}}{\vdash\Gamma, \Delta, a \lpar (b\ltens c)}{\vdash\Gamma, a,b}{\vdash\Delta, c}
			\mbox{\quad and \quad}
			\vliinf{}{\mathcal{R}_{H_2}}{\vdash\Gamma, \Delta, a \lpar (b\ltens c)}{\vdash\Gamma, a,c}{\vdash\Delta, b}	
			\\\hline&&\\[-8pt]
			\mbox{Associated partitions}
			&
			\part{\block{1},\block{2,3}}
			&
			\part{\block{1,2},\block{3,4}}
			&
			\part{\block{1,2},\block{3}}
			\mbox{\quad and \quad}
			\part{\block{1,3},\block{2}}
		\end{array}
	$}\end{equation*}
\end{example}

Conversely, 
given a set of partitions $P$ in $\partofn$, we can define a rule introducing an $n$-ary \defn{generalized connective} $\con_P$ 
for each partition in $P$.
In this case, we say that $P$ is the \defn{behavior} of $\con_P$.
Consider the following examples.
\begin{equation}\label{eq:Gir4}
\adjustbox{max width = .9\textwidth}{$
	\begin{array}{c@{\quad}|@{\quad}c@{\qquad}|@{\qquad}c}
		\mbox{Behavior}
		&
		P=\Set{\part{\block{1,2},\block{3,4}},\block{1,4},\block{2,3}}
		&
		Q=\Set{\part{\block{1,3},\block2,\block4},\part{\block1,\block{2,4},\block3}}
		\\\hline&&\\
		
		&
		\vliinf{}{\mbox{\scriptsize$\part{\block{1,2},\block{3,4}}$}}{\vdash \Gamma,\Delta ,\con_P(A,B,C,D)}{\vdash \Gamma, A, B}{\vdash \Delta, C,D}
		&
		\vliiinf{}{\mbox{\scriptsize$\part{\block{1,3},\block2,\block4}$}}{\vdash \Gamma,\Delta,\Sigma ,\con_Q(A,B,C,D)}{\vdash \Gamma, A, C}{\vdash \Delta,B}{\vdash \Sigma, D}
		\\\mbox{Rules}&&\\
		&
		\vliinf{}{\mbox{\scriptsize$\part{\block{1,4},\block{2,3}}$}}{\vdash \Gamma,\Delta ,\con_P(A,B,C,D)}{\vdash \Gamma, A, D}{\vdash \Delta, B,C}
		&
		\vliiinf{}{\mbox{\scriptsize$\part{\block1,\block{2,4},\block3}$}}{\vdash \Gamma,\Delta,\Sigma ,\con_Q(A,B,C,D)}{\vdash \Gamma, A}{\vdash \Delta,B, D}{\vdash \Sigma, C}
	\end{array}
 $}
\end{equation}

However, as shown in various works \cite{dan:reg:89,gir:meanII,mai:14,acc:mai:20}, not all sets of partitions can be considered to be satisfactory in order  to define connectives.
In fact, we allow to use a set of partitions $P\subset\partofn$ to describe a connective only if $P$ admits a set $Q$ such that $P\perp Q$. 
This condition is mandatory to guarantee the possibly to define a dual connective whose rules well-behave with respect to cut-elimination.

In \cite{mai:14,acc:mai:20} it has been proved that there are families of sets of partitions which can be used to describe behaviors different from any synthetic rule defined using $\ltens$ and $\lpar$ rules.
Moreover, in \cite{acc:mai:20} it is also shown that no satisfactory sequent calculus can be defined in presence of generalized connectives due to the lack of the so-called \defn{initial coherence}~\cite{dale-elaine,avrron:canonical} (also called \emph{packaging problem} in \cite{dan:reg:89}), that is, the possibility of having a proof system in which it is always possible to prove ``$A$ implies $A$'' using atomic axiom only.

\subsection{Generalized Connectives in Multiplicative Proof Nets}

Sets of partitions have been used to define generalized connectives in the \mllpn syntax in \cite{dan:reg:89,gir:meanII,mai-pui05,mai:14,acc:mai:20}, overcoming the aforementioned problem of initial coherence.
We here give some intuitions on these connectives, while more precise definitions are provided in \Cref{sec:gen} where we properly define the formal setting required to accommodate them.

Generalized connectives in multiplicative proof structures use sets of partitions to define the \emph{behavior} of new connectives, that is, the way \emph{tests} are constructed.
Intuitively, 
the behavior of $\lpar$ (defined as $\set{\part{\block{1},\block{2}}}$) and $\ltens$ (defined as $\set{\part{\block{1,2}}}$) provide the topological constraints of the definition of the test: for $\lpar$ the link is replaced by a hyperedge connecting only one of the two inputs (connecting the output to one of the two blocks) while for the $\ltens$ the link is replaced by hyperedge connecting both inputs (since both belong to the same block).
Similarly, in defining a test for a link with a given behavior is replaced by certain hyperedges connecting the vertices in a same block.

In this case, given an $n$-ary connective and a partition $P\subset\partofn$,
the condition of the existence of a $Q$ such that $P\perp Q$ is not enough to guarantee the existence of a dual connective well-behaving with respect to cut-elimination.
Thus the stronger condition $P\dperp Q$ is needed%
\footnote{
	Note that in \cite{dan:reg:89,gir:meanII} each multiplicative connective is defined by a pair of sets of dual partitions over the same finite set satisfying an orthogonality condition.
	This condition is sufficient to fully describe these connectives in sequent calculus style, and we here show that it is also sufficient for our proof net expansion paradigm.
	However, it is well-known that in a Curry-Howard oriented interpretation of proof-as-program paradigm stronger conditions are required in order to guarantee a sound dynamic of cut-elimination (that is, not only the two partitions must be orthogonal, but also their orthogonal sets must be so). 
}%
.

\begin{remark}
	As noticed in \cite{mai:14,acc:mai:20}, 
	the definition of more-than-binary generalized connectives 
	requires to include the information about which
	block of inputs is connected to the output of the link connective.
	This information is only required to ensure a sound cut-elimination procedure,
	and it is lost after removing cuts.
    Said differently, the contextual equivalence defined by cut-elimination is not able to distinguish certain behaviors differing in the way set of inputs are connected to the output.
    This information is not relevant for the standard $\MLL$ connectives nor for  \emph{synthetic connectives} (that is, the ones which can indirectly defined by means of combination of $\ltens$ and $\lpar$; see \Cref{ex:synthetic}), since it can be indirectly derived using the less complex nature of these connectives, which are defined inductively using binary ones.
	Nevertheless, this information is not negligible in the general case, since this information may define different tests as shown in the following example explaining in detail \Cref{fig:synthTests}.
\end{remark}

\begin{figure}[t]\label{ex:synthTests}
    \newgate{con}{\widetilde{\kappa}}{}
    \adjustbox{max width=.92\textwidth}{$
        \begin{array}{c@{\quad}|@{\quad}c@{\quad}|@{\quad}c}
            \mbox{Formula} & \mbox{Proof Structure} & \mbox{Tests}
            \\
            \begin{array}{c}
                \kappa(a,b,c)
                \\
                =
                \\
                a\ltens(b\lpar c)
            \end{array}
            &
            \begin{array}{ccccc}
                &\vb1 && \vc1
                \\[1em]
                \va1 &&\Gpar1
                \\[1em]
                &\Gtens1
                \\
                &\vfor1{\kappa(a,b,c)}
            \end{array}
            \pswires{a1/Gtens1.L,b1/Gpar1.L,c1/Gpar1.R,Gpar1.O/Gtens1.R,Gtens1.O/for1}
            &
            \begin{array}{c|c}
                \begin{array}{ccccc}
                    &\vb1 && \vc1
                    \\[1em]
                    \va1 &&\vfor1{b\lpar c}
                    \\
                    &\vfor2{\kappa(a,b,c)}
                \end{array}
                \Gedges{a1/for2,b1/for1,for1/for2}
                &
                \begin{array}{ccccc}
                    &\vb1 && \vc1
                    \\[1em]
                    \va1 &&\vfor1{b\lpar c}
                    \\
                    &\vfor2{\kappa(a,b,c)}
                \end{array}
                \Gedges{a1/for2,c1/for1,for1/for2}
            \end{array}
            \\\hline
            \begin{array}{c}
                \cneg \kappa(\cneg a,\cneg b,\cneg c)
                \\
                =
                \\
                \cneg a\lpar(\cneg b\ltens \cneg c)
            \end{array}
            &
            \begin{array}{ccccc}
                &\vnb1 && \vnc1
                \\[1em]
                \vna1 &&\Gtens1
                \\[1em]
                &\Gpar1
                \\
                &\vfor1{\cneg \kappa(a,b,c)}
            \end{array}
            \pswires{na1/Gpar1.L,nb1/Gtens1.L,nc1/Gtens1.R,Gtens1.O/Gpar1.R,Gpar1.O/for1}
            &
            \begin{array}{c|c}
                \begin{array}{ccccc}
                    &\vnb1 && \vnc1
                    \\[1em]
                    \vna1 &&\vfor1{b\ltens c}
                    \\
                    &\vfor2{\cneg \kappa(a,b,c)}
                \end{array}
                \Gedges{na1/for2,nb1/for1,nc1/for1}
                &
                \begin{array}{ccccc}
                    &\vnb1 && \vnc1
                    \\[1em]
                    \vna1 &&\vfor1{b\ltens c}
                    \\
                    &\vfor2{\cneg \kappa(a,b,c)}
                \end{array}
                \Gedges{nb1/for1,nc1/for1,for1/for2}
            \end{array}
            \\\hline
            \widetilde{\kappa}(a,b,c)
            &
            \begin{array}{ccc}
                \va1&\vb1 & \vc1
                \\[1em]
                &\Gcon1
                \\[1em]
                &\vfor1{\widetilde \kappa(a,b,c)}
            \end{array}
            \pswires{a1/Gcon1.L,b1/Gcon1.north,c1/Gcon1.R,Gcon1.O/for1}
            &
            \begin{array}{c|c|c}
                \begin{array}{ccc}
                    \va1&\vb1 & \vc1
                    \\[1em]
                    &\vfor1{\widetilde \kappa(a,b,c)}
                \end{array}
                \Gedges{a1/for1,b1/for1}
                &
                \begin{array}{ccc}
                    \va1&\vb1 & \vc1
                    \\[1em]
                    &\vfor1{\widetilde \kappa(a,b,c)}
                \end{array}
                \Gedges{a1/for1,c1/for1}
                &
                \begin{array}{ccc}
                    \va1&\vb1 & \vc1
                    \\[1em]
                    &\vfor1{\widetilde \kappa(a,b,c)}
                \end{array}
                \Gedges{a1/b1,c1/for1}
            \end{array}
        \end{array}
        $}
    \caption{
        The behaviors of the synthetic connectives 
        associated to the formula 
        $F(a,b,c)= a \ltens (b \lpar c)$,
        to its dual formula,
        and 
        to a generalized multiplicative connective whose behavior strictly contains the one of $F(a,b,c)$ 
        in such a way is the same of $F(a,b,c)$ if restricted on the inputs only.
    }
    \label{fig:synthTests}
\end{figure}

\begin{example}
    Consider the formula
    $F(a,b,c)= a \ltens (b \lpar c)$ 
    and the synthetic connectives  
    $\kappa(a,b,c)\coloneqq F(a,b,c)$  and  
    $\cneg\kappa(\cneg a, \cneg b, \cneg c)\coloneqq \cneg F(\cneg a, \cneg b, \cneg c)$
    respectively associated to the sets of partitions 
    $\behof\kappa=\set{\part{\block{1,2},\block 3},\part{\block{1,3},\block 2}}$
    and
    $\behof{\cneg \kappa}=\set{\part{\block{1},\block{2,3}}}$ (see \Cref{fig:synthTests}).
    
    We can now define the connective $\widetilde \kappa$ associated to the same set of partitions of $\kappa$ (that is, the set of partitions $\behof{\widetilde \kappa}=\behof\kappa=\set{\part{\block{1,2},\block 3},\part{\block{1,3},\block 2}}$ ) but in which we allow an extra test which enforces no new partitions among the inputs.
    See the bottom-most row of the table in \Cref{fig:synthTests}, where the new test (the right-most one) enforces the partition $\part{\block{1,2},\block 3}\in \behof\kappa$ 
    over inputs.

    Since $\kappa$ and $\widetilde{\kappa}$ are defined by the same set of partitions over their inputs, 
    they are both orthogonal to $\cneg \kappa$.
    Moreover, both \DRcor \mllpss of $\kappa(a,b,c)\lpar \cneg\kappa(\cneg a,\cneg b, \cneg c)$ and $\widetilde\kappa(a,b,c)\lpar \cneg\kappa(\cneg a,\cneg b, \cneg c)$ are correct, 
    and the result of cut-elimination of a $\kappa$- or a $\widetilde{\kappa}$-gate against a $\cneg\kappa$-gate reduces to a \mllps with the same behavior.
    Note that this implies that $\kappa$ and $\widetilde{\kappa}$ are indistinguishable with respect to the notion of context equivalence usually considered on proof structures (see, e.g., \cite{gir:trans, eng:sei:stellar,eng:sei:trans}).
\end{example}

\section{Logic Programming with Multiplicative Structures}\label{sec:logProg}
\def\bipoleset{\mathfrak F}

In this section we recall the results from \cite{and:02,and:maz:03} (restricted to the multiplicative linear logic fragment) on the possibility to define a logic programming framework based on proof net construction.


The classical interpretation of logic programming (see, e.g., \cite{andreoli1992logic,and:02,dale-elaine}),
a program is defined by a set of sequent calculus rules and 
its execution is conceived as the process of expanding the open branches of the derivation tree of a given formula.
This correspondence can easily be extended using synthetic (linear) inference rules as the ones from \Cref{ex:synthetic} to define the following methods:
\begin{equation}\label{eq:synthRulesMethods}
	\begin{array}{c@{\quad}|@{\quad}c@{\quad}|@{\quad}c}
		\prolog{F\colonminus a , (b\lpar c)}
		&
		\prolog{G \colonminus (a\lpar b) , (c\lpar d)}
		&
		\prolog{H_1 \colonminus (a \lpar b), c}
		\qquad
		\prolog{H_2 \colonminus (a \lpar c), b}
	\end{array}
\end{equation}

In particular, a specific family of formulas (called \emph{bipoles}) can be used to define methods.
\begin{definition}[\cite{and:02}]
	Given a set of \defn{negative} atoms $\atoms$ whose negations are \defn{positive}, a \defn{monopole} is a disjunction ($\lpar$) of negative atoms.
	A \defn{bipole} is a conjunction ($\ltens$) of monopoles and  positive atoms which contains at least one positive atom.
	Given a set of bipoles $\bipoleset$, the \defn{focussing bipolar sequent calculus} $[\bipoleset, \atoms]$ is given by the set of inference rules of the following form, where $F$ is a bipole in $\bipoleset$.
	\begin{equation}\label{eq:bipole}
		\vlinf{}{\mbox{\scriptsize$F=\outp_1^\perp \ltens \cdots \ltens \outp_n^\perp \ltens (\inp_{1,1}\lpar \cdots\lpar \inp_{1,k_1})\ltens\cdots\ltens  (\inp_{i,1}\lpar \cdots\lpar \inp_{i,k_i})$}}{\vdash \outp_1, \dots, \outp_n}{\vdash \inp_{1,1}, \dots, \inp_{1,k_1} \qquad\cdots\qquad \vdash \inp_{i,1}, \dots, \inp_{i,k_i}}
	\end{equation}
\end{definition}

As shown in \cite{andreoli1992logic}, a bipole $F$ can be seen as a logic programming method having as 
\defn{head} (or \defn{trigger}) the subformula containing the positive atoms of $F$ (a conjunction), and 
as \defn{body} the subformula containing the negative atoms $F$ (a \CNF formula).
Intuitively, each bipole $F$ induces a synthetic rule with principal formula $F$ and 
and whose active formulas are its positive atoms gathered in a same premise if they belong to the same conjunct.
By means of example the rule for the bipole $F$ in \Cref{eq:bipole} can be seen as a synthetic rule introducing the formula $F$ corresponding to the following derivation
\begin{equation}\label{eq:bipoleSynth}
	\vlderivation{
		\vliin{}{\ltens}{\outp_1^\perp \ltens \cdots \ltens \outp_n^\perp \ltens (\inp_{1,1}\lpar \cdots\lpar \inp_{1,k_1})\ltens\cdots\ltens  (\inp_{i,1}\lpar \cdots\lpar \inp_{i,k_i}), \outp_1, \cdots ,\outp_n}{
			\vliiiq{}{\ltens}{(\inp_{1,1}\lpar \cdots\lpar \inp_{1,k_1})\ltens\cdots\ltens  (\inp_{i,1}\lpar \cdots\lpar \inp_{i,k_i})}{
				\vliq{}{\lpar}{\inp_{1,1}\lpar \cdots\lpar \inp_{1,k_1}}{\vlhy{\inp_{1,1}, \dots, \inp_{1,k_1}}}
			}{\vlhy{\cdots}}{
				\vliq{}{\lpar}{\inp_{i,1}\lpar \cdots\lpar \inp_{i,k_i}}{\vlhy{\inp_{i,1}, \dots, \inp_{i,k_i}}}
			}
		}{
			\vliiiq{}{\ltens}{\outp_1^\perp \ltens \cdots \ltens \outp_n^\perp , \outp_1, \cdots ,\outp_n}{\vlin{}{\axrule}{\vdash\outp_1^\perp,\outp_1}{\vlhy{}}}{\vlhy{\cdots}}{\vlin{}{\axrule}{\vdash\outp_n^\perp, \outp_n}{\vlhy{}}}
		}
	}
\end{equation}
In \cite{and:01} it has been proved that the focussing bipolar sequent system with one rule for each $\MLL$ bipole is sound and complete with respect to $\MLL$.

\subsection{Bipolar Proof Nets}

\begin{figure}[t]
	\pzvertex{bb}{b'\lpar b''}{}
	
	\resizebox{\textwidth}{!}{
		\begin{tabular}{c|c|c|c}
			Bipoles								&
			Sequent calculus derivation with	&
			focussing bipolar expansion			& 
			Proof Net expansion 
			\\[-5pt]
			&
			concurrent expansion ($G$ and $H$) 	&
			(concurrent synthetic rules) 		&
			(one link of type $G\parallel H$)
			\\[-5pt]
			$
			\begin{array}{l}
				F=\cneg a \ltens b \ltens c \\
				G=\cneg b \ltens (b'\lpar b'')\\
				H=\cneg c \ltens c'
			\end{array}
			$
			&
			$\vlderivation{\vliin{}{\mathcal R_F}{F,a,G, H}{\vlhy{b,G}}{\vlhy{c,H}}}$
			$\;\rightsquigarrow\;$
			$\vlderivation{\vlin{}{\lpar}{F,a,G\lpar H}{\vliin{}{\mathcal R_F}{F,a,G,H}{\vlin{}{\mathcal R_G}{b,G}{\vlhy{b'\lpar b''}}}{\vlin{}{\mathcal R_H}{c,H}{\vlhy{c'}}}}}$
			&
			$\vlderivation{\vliin{}{F}{a}{\vlhy{b}}{\vlhy{c}}}$
			$\;\rightsquigarrow\;$
			$\vlderivation{\vliin{}{F}{a}{\vlin{}{G}{b}{\vlhy{b'\lpar b''}}}{\vlin{}{H}{c}{\vlhy{c'}}}}$
			&
			$
			\begin{array}{c@{\hspace{-10pt}}c@{\hspace{-10pt}}c}
				\vb1 & &\vc1 \\[1em]
				& \vmod1{\qquad F \qquad}
				\\[.5em]
				&\va1
			\end{array}
			\pswires{b1/mod1.150,c1/mod1.30,mod1.-90/a1}
			\;\rightsquigarrow\;
			\begin{array}{c@{\hspace{-10pt}}c@{\hspace{-10pt}}c}
				\\
				\vbb2 & &\vc2' \\[.5em]
				& \vmod2{\qquad G\parallel H \qquad}
				\\[.5em]
				\vb1 & &\vc1 \\[.5em]
				& \vmod1{\qquad F \qquad}
				\\[.5em]
				&\va1
			\end{array}
			\pswires{b1/mod1.150,c1/mod1.30,mod1.-90/a1}
			\pswires{bb2/mod2.150,c2/mod2.30}
			\pswires{mod2.-150/b1,mod2.-30/c1}
			$
		\end{tabular}
	}
	\vspace{-5pt}
	\caption{A concurrent application of the bipoles $G$ and $H$ after $F$ represented in different formalisms}
	\label{fig:concurrentEx}
\end{figure}

The idea of using the focussing bipolar sequent calculus has been further developed in \cite{and:maz:03}, where the authors proposed to model such a framework using  proof nets construction instead of proof search in sequent calculi.
The main advantage of the graphical syntax with respect to the bipolar sequent calculus is that in this latter, even if this rule admits a non-singleton trigger, a rule can expand only a single branch of a derivation.
In fact, the tree-like structure of sequent calculus syntax allows us to expand one leaf of the derivation tree at a time by applying a rule.
For instance, consider \Cref{fig:concurrentEx} where the concurrent application of two methods on two different branches of a derivation can be represented by an expansion of a \mllps with a single link.

More precisely, we can use bipoles to define new link types in the same manner a bipole defines a new synthetic inference rule in the sequent calculus.
Each test replaces  such a link with one of the possible tests of the \mllps representing the bipole.
By means of example, the bipolar rule from \Cref{eq:bipole} could be used to define the link below on the left, 
and tests would replace such a link with any test of the \mllps below on the right, which represents the open derivation in \Cref{eq:bipoleSynth}.
%
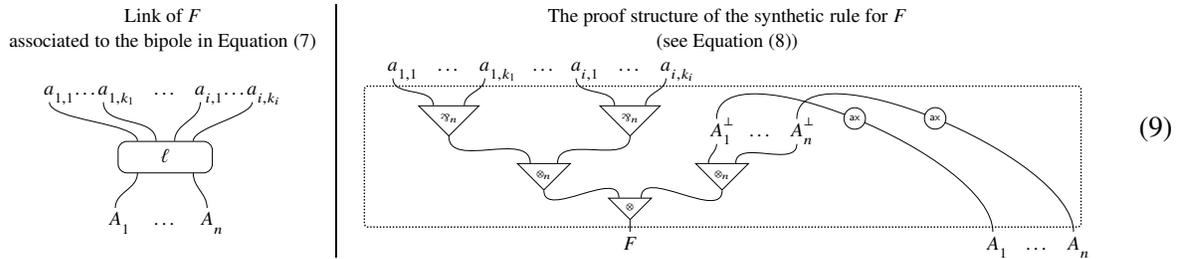
\begin{equation}\label{eq:bipolarLink}
	\adjustbox{max width=.92\textwidth}{$\begin{array}{c@{\quad}|@{\quad}c}
			\mbox{Link of $F$}
			&
			\mbox{The \mllps of the synthetic rule for $F$}
			\\
			\mbox{associated to the bipole in \Cref{eq:bipole}}
			&
			\mbox{(see \Cref{eq:bipoleSynth})}
			\\
			\begin{array}{c@{}c@{}c@{\hspace{-10pt}}c@{\hspace{-10pt}}c@{}c@{}c}
				\va1_{1,1} &\dots & \va2_{1,k_1} & \dots &  \va3_{i,1} &\dots & \va4_{i,k_i} \\[2em]
				&       &           & 	\vmod1{\qquad\link\qquad}\\[2em]
				&&      ~\vA1_1      & \dots &~\vA2_n
			\end{array}
			\pswires{a2/mod1.120,a3/mod1.60}
			\pslwire{a1}{mod1.150}{.5}{}
			\pslwire{a4}{mod1.30}{.5}{}
			\pswires{mod1.-150/A1,mod1.-30/A2}
			&
			\begin{array}{cc@{}c@{}c@{}c@{}c@{}c@{}c@{}c@{}ccc@{\qquad\qquad\qquad}cccccc}
				&\va1_{1,1} &\dots & \va2_{1,k_1} & \dots &  \va3_{i,1} &\dots & \va4_{i,k_i} \\[-1em]
				\virt1
				\\[1em]
				& & \Gmpar1 &           & 	    &           &\Gmpar2&  &
				\vA1_1^\perp &\dots& \vA3_n^\perp & 
				\\[1em]
				&&&&\Gmtens2&&&&\Gmtens1
				\\
				&&&&&&\Gtens3 
				\\
				&&&&&&&&&&&&&&&&\virt2
				\\[-1em]
				&&&&&&\vF1&&&&&&&&\vA2_1 & \dots & \vA4_n
			\end{array}
			\pswires{A1/Gmtens1.L,A3/Gmtens1.R}
			\pswires{a1/Gmpar1.L,a2/Gmpar1.R,a3/Gmpar2.L,a4/Gmpar2.R}
			\pswires{Gmpar1.O/Gmtens2.L,Gmpar2.O/Gmtens2.R}
			\pswires{Gmtens2.O/Gtens3.L, Gmtens1.O/Gtens3.R}
			\psaxiom{A1}{A2}{.65}{}
			\psaxiom{A3}{A4}{.65}{}
			\linkbox{irt1}{irt2}
			\pswires{Gtens3.O/F1}
		\end{array}$}
\end{equation}

Note that the link above on the left has outputs $\outp_1,\ldots, \outp_n$, while the \mllps on the right has an additional output $F$. 
In the next section we provide a solution to address this mismatch (see no-output links in \Cref{def:links}), but it is worth noting
that we can define links representing concurrent application of bipoles by simply connecting those additional outputs via a $\lpar$-link.
Analyzing the shape of the \mllps describing a \emph{concurrent bipole}.

\begin{definition}\label{def:concBip}
	We introduce the following naming for specific proof structure (see examples in \Cref{fig:bipLinkDec}):
	
	\begin{itemize}
		\item \defn{body}: a $\lmtens$-link collecting the outputs of $\lmpar$-links (i.e., the \mllps of a \CNF-formula). 
		The body gathers the clauses corresponding to the body of a method;
		
		\item \defn{header}: a bundle of $\axrule$-links attached to a $\lmtens$-link by exactly one of their two outputs each.
		The header gathers the outputs corresponding to the head of a method;
		
		\item \defn{synchronizer}: a $\lmpar$-link collecting the outputs of $\ltens$-links (i.e., the \mllps of a \DNF-formula).
		The synchronizer establishes a connection between headers of methods and their bodies.
	\end{itemize}
	A \defn{concurrent bipole} is a proof structure made of a synchronizer whose inputs are attached to headers and bodies.
	That is, a \mllps of the following shape:
	\begin{equation}\label{eq:bipLink}
		\adjustbox{max width=.9\textwidth}{$
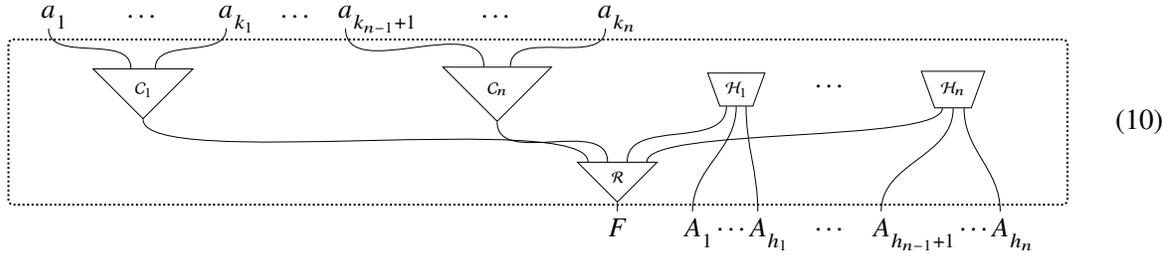
			\begin{array}{cccccccccccc}
				&
				\va1_1 & \dots & \va2_{k_1} & \dots & \va3_{k_{n-1}+1} & \dots & \va4_{k_n}
				\\[-1em]
				\virt1&
				\\[1em]
				&& \vGmod1{\quad\cmod_1\quad} &  &  &  &\vGmod2{\quad\cmod_n\quad}  & &  \vTmod4{\hmod_1} & \cdots & \vTmod5{\hmod_n} 
				\\[2em]
				&&&&&&&\vGmod3{\quad\rmod\quad}
				\\[5pt]
				&&&&&&&&&&&\virt2
				\\[-1em]
				&&&&&&&\vF1& \vA1_1 \cdots \vA2_{h_1}&\cdots &\vA3_{h_{n-1}+1} \cdots \vA4_{h_n}
			\end{array}
			\pswires{a1/Gmod1.120,a2/Gmod1.60,a3/Gmod2.120,a4/Gmod2.60}
			\pswires{Gmod2/Gmod3.120,Tmod4.-120/Gmod3.60}
			\pswire{Gmod1}{Gmod3.150}{.3}{}
			\pswire{Tmod5.-120}{Gmod3.30}{.3}{}
			\pswires{Tmod4/A1,Tmod4.-60/A2,Tmod5/A3,Tmod5.-60/A4}
			\linkbox{irt1}{irt2}
			\pswires{Gmod3.O/F1}
			$}
	\end{equation}
	where 
	$\hmod_1,\ldots,\hmod_n$ are headers,
	$\cmod_1,\ldots,\cmod_n$ are bodies, and
	$\rmod $ is a synchronizer.

\end{definition}

\begin{figure}[t]
	\resizebox{\textwidth}{!}{$
		\begin{array}{c|c|c}
			\mbox{Body}
			&
			\mbox{Header}
			&
			\mbox{Synchronizer}
			\\\hline&&\\
			\begin{array}{ccccccc}
				\va1_1 & \dots &\va2_{i_1} & \dots & \va3_{i_{k-1}+1} & \dots & \va4_{i_k}
				\\[1em]
				&\Gmpar1&           &        &                   &\Gmpar2
				\\[1em]
				&       &           &\Gmtens1&                   &
			\end{array}    
			\pswires{a1/Gmpar1.L,a2/Gmpar1.R,a3/Gmpar2.L,a4/Gmpar2.R}
			\pswires{Gmpar1.O/Gmtens1.L,Gmpar2.O/Gmtens1.R}
			&
			\begin{array}{ccccccccc}
				\\[1em]
				\vA1_1^\perp  & \dots & \vA3_n^\perp  & &\vA2_1&\dots& \vA4_n
				\\[1em]
				&\Gmtens1
			\end{array}    
			\pswires{A1/Gmtens1.L,A3/Gmtens1.R}
			\psaxioms{A1/A2,A3/A4}
			&
			\begin{array}{ccccccc}
				\va1_1 &  &\va2_{2} & \dots & \va3_{2k-1} &  & \va4_{2k}
				\\[1em]
				&\Gtens1&           &        &                   &\Gtens2
				\\[1em]
				&       &           &\Gpar1&                   &
			\end{array}    
			\pswires{a1/Gtens1.L,a2/Gtens1.R,a3/Gtens2.L,a4/Gtens2.R}
			\pswires{Gtens1.O/Gpar1.L,Gtens2.O/Gpar1.R}
		\end{array}
		$}
	\caption{Components of a concurrent bipolar link.}
	\label{fig:bipLinkDec}
\end{figure}

\section{Generalizing Multiplicative Proof Structures}\label{sec:gen}
\def\multstr{multiplicative structure\xspace}
\def\multstrs{multiplicative structures\xspace}

In this section we provide a general setting to define hypergraphs with hyperedges labeled by sets of partitions generalizing the syntax of \mllpss, allowing us to accommodate generalized connectives, generalize the \DRcor, and define the notion of a \emph{component} as a ``proof structure which may be a piece of a proof net''.

\begin{definition}\label{def:links}
	A \defn{link type} (or simply \defn{type}) is a triple
	$\tuple{n, m , \behsym}$
	given by two natural numbers $n,m\in\N$ and a \defn{behavior}
	$\behsym\subseteq\partof{\set{\inp_1,\ldots,\inp_n,\outp_1,\ldots, \outp_m}}$.
    We define the following link types:
    $$
    \adjustbox{max width = \textwidth}{$
    \begin{array}{c@{\;}c@{\;}l@{\quad}c@{\;=\;}l@{\quad}c@{\;=\;}l}
        \axrule 	&=& \tuple{ 0, 2 , \Set{\part{\block{\outp_1, \outp_2}}}}
        &
        \ltens 		& \tuple{ 2, 1 , \Set{\part{\block{\inp_1, \inp_2,\outp_1}}}}
        &
        \lpar 		& \tuple{ 2, 1 , \Set{\part{\block{\inp_1,\outp_1},\block{\inp_2}},\part{\block{\inp_1},\block{\inp_2,\outp_1}}}}
        \\
        \cutr 		&=& \tuple{ 2, 0  ,\part{\block{\outp_1,\outp_2}}}
        &
        \lmparp 	& \tuple{ n, 0  ,\part{\block{\outp_1},\dots,\block{\outp_n}}}
        &
        \lmtensp 	& \tuple{ n, 0  ,\part{\block{\outp_1,\ldots, \outp_n}}}
        \\
        &&&
        \lmtens 	& \tuple{ n, 1  ,\Set{\part{\block{\inp_1,\ldots, \inp_n,\outp_1}}}}
        &
        \lmpar 		& \tuple{ n, 1  , \Set{\part{\block{\inp_1}, \dots, \block{\inp_{k-1}}, 
                    \block{\inp_{k},\outp_1},
                    \block{\inp_{k+1}},\dots,\block{\inp_n}}}_{k\in\intset1n}}
    \end{array}
    $}
    $$
\end{definition}
\begin{remark}
	By definition, 
	$\ltens_2=\ltens$ and $\lpar_2=\lpar$ and $\cutr=\ltens^\bullet_2$.
	The type
	$\ltens_1=\lpar_1$ 
	can be thought as an edge connecting the input with the output
	since they both have behavior $\part{\block{\inp_1,\outp_1}}$.
    The type $\ltens_1^{\bullet}=\lpar_1^{\bullet}$ can be thought as a ``dead-end'' hyperedge with one input and no output (their behavior is $\part{\block{\inp_1}}$).
\end{remark}

\begin{definition}
    A \defn{\multstr} over a signature $\Sigma$ is a linear hypergraph $\gH$ such that each hyperedge $\link$ is labelled with a $\tuple{n,m,\behsym}\in\Sigma$ such that $\sizeof{\iof\hed}=n$ and $\sizeof{\oof\hed}=m$.
	
    In drawing \mss, we label hyperedges by the corresponding type.
    The definition of \defn{sub-\ms}, as well as the definition of \defn{sequential} and \defn{parallel composition} of \mss are defined extending the ones for hypergraphs.
\end{definition}

In order to extend the notion of \DRcor, we need to provide a way to define tests.
\begin{definition}
	Let $\gS=\tuple{\vertices[\gG],\hedge[\gS]}$ be a \ms.
	A \defn{switching} for $\gS$ is a map $\switch$ assigning to each link $\link$ a single partition $\switch(\link)\in \behof\link$.
	We denote by $\switchset{\gS}$ the set of all possible switchings for $\gS$.
	The \defn{test} $\switch(\gS)$ defined by the switching $\switch$
	is the undirected hypergraph obtained by replacing each link in $\hedge[\gS]$ with one undirected hyperedge for each block in $\switch(\link)$.
	Formally, 
	$\switch(\gS)=\tuple{\vertices[\gS] \;,\; \bigcup_{\link \in \hedge[\gS]}\set{\gamma \mid \gamma \in \switch(\link)}}$.
	
	The \defn{behavior} of a test $\switch$ is the partition $p_\switch$ of the border of $\gS$ defined by the connected components of $\switch(\gS)$.
	That is, $x,y\in\bordof\gS$ belongs to the same block $\gamma \in p_\switch$ iff the vertices $x$ and $y$ are connected of $\switch(\gS)$.
	The \defn{behavior} of a \ms $\gS$ is defined as the set of behaviors of its tests, that is, $\behof\gS=\Set{p_\switch\in\partof{\bordof{\gS}}\mid\switch\in\switchset{\gS}}$.
\end{definition}

The behavior of a \ms is the collection of the information on how a \ms interacts with any possible context.
It keeps track of the connectivity the vertices in its border in all its tests (see an example in \Cref{fig:extest}).

\begin{figure}[t]
	\newgate{con}{\kappa}{}
	\adjustbox{max width=\textwidth}{$\begin{array}{c|c}
			\begin{array}{ccccccccc}
				\va1 &\vb2 &\vc3 & \vd4 &&\ve5
				\\[.5em]
				&\Gcon1&& &\Gtens1
				\\[1em]
				&&&\Gpar1
				\\
				&&&\vo0
			\end{array}
			\pswires{a1/Gcon1.L,b2/Gcon1.north,c3/Gcon1.R}
			\pswires{d4/Gtens1.L,e5/Gtens1.R}
			\pswires{Gcon1.O/Gpar1.L,Gtens1.O/Gpar1.R,Gpar1.O/o0}
			&
			\begin{array}{c|c}
				\switch(\kappa)=\part{\block{\outp_1,\inp_1,\inp_2},\block \inp_3}
				\quad
				\switch(\ltens)=\part{\block{\op_1,\ip_1,\ip_2}}
				\quad
				\switch(\lpar)=\part{\block{\op_1,\ip_1},\block{\ip_2}}
				&
				\switch(\kappa)=\part{\block{\outp_1,\inp_1,\inp_3},\block{\inp_2}}
				\quad
				\switch(\ltens)=\part{\block{\op_1,\ip_1,\ip_2}}
				\quad
				\switch(\lpar)=\part{\block{\op_1,\ip_2},\block{\ip_1}}
				\\\hline\\
				\begin{array}{ccccccccc}
					\va1 &\vb2 &\vc3 & \vd4 &&\ve5
					\\[1em]
					&\vfor1{\kappa(a,b,c)}&& &\vfor2{d\ltens e}
					\\[1em]
					&&&\vfor3{\kappa(a,b,c)\lpar(d\ltens e)}
				\end{array}
				\Gedges{a1/for1,b2/for1,d4/for2,e5/for2,for1/for3}
				&
				\begin{array}{ccccccc@{\hspace{-5pt}}c@{\hspace{-5pt}}ccc}
					\va1 &\vb2 &\vc3 & \vd4 &&\ve5
					\\[1em]
					&\vfor1{\kappa(a,b,c)}&& &\vfor2{d\ltens e}
					\\[1em]
					&&&\vfor3{\kappa(a,b,c)\lpar(d\ltens e)}
				\end{array}
				\Gedges{a1/for1,c3/for1,d4/for2,e5/for2,for2/for3}
			\end{array}
		\end{array}$}
	\caption{A \mllps where 
		$\kappa
		=
		\tuple{3,1,\set{
				\part{\block{\outp_1,\inp_1,\inp_2},\block \inp_3},
				\part{\block{\outp_1,\inp_1,\inp_3},\block{\inp_2}}
			}
		}$
		and two of its test}
	\label{fig:extest}
\end{figure}

\begin{definition}\label{def:netGCor}
	Let $\gS$ be a \ms.
	We say that $\gS$ is \defn{correct} if $\switch(\gS)$ is connected and acyclic for any test $\switch\in\switchset{\gS}$.
	It is \defn{\mnet} if correct and if $\gS$ has no inputs and at least one output.
	If each test $\switch(\gS)$ of $\gS$ is acyclic and each of its vertices is connected to a vertex of the border, then we say that $\gS$ is a \defn{(multiplicative) \mcomp}.
	A \defn{transitory} component (or \defn{\tracomp}) is a \mcomp such that each input admits a test where it is connected to an output.
	A \defn{\mmod} $\mod$ is a connected non-empty \ms such that 
	$\mod\subset \gS$ for a \mn $\gS$.
\end{definition}
\begin{example}
	Consider the examples in \Cref{fig:modEx}.
	The \mstr $\prs_1$ is a \mnet (therefore a \tracomp).
	$\prs_2$ is a \mcomp and a \mmod, but not a \tracomp (it has no outputs).
	The \ms $\prs_3$ is not a \mcomp (it admits a test in which $a_2$ and $a_3$ are disconnected from any other vertex)
	nor a \mmod (there is no \mn containing $\prs_3$ as a sub-\ms).
\end{example}

\begin{theorem}\label{thm:Mod:char}
	All \mmods are \mcomps.
\end{theorem}
\begin{proof}
	If $\mod$ is a \mmod, then 
	each test $\switch(\mod)$ is acyclic, otherwise  there should be a cycle in a test of $\gS$.
	Moreover, 
	as consequence of the fact that $\gS$ is connected,
	no sub-\ms $\gS'$ of $\gS$ such that  $\bordof{\gS'}=\emptyset$.
	Therefore each vertex in $\mod$ must be connected to a vertex in $\bordof{\mod}$ in each test $\switch(\mod)$ otherwise $\gS$ would not be connected.
\end{proof}

\begin{definition}
	Let $\mod_1$ and $\mod_2$ be \mcomps
	and
	$X \subseteq \left(\inpof{\mod_1}\cap \outpof{\mod_2}\right) $ non-empty.
	We say that $\mod_2$ \defn{expands} $\mod_1$  (on $X$) 
	if $\mod_1\comp X \mod_2$ is a \mcomp.
\end{definition}

\begin{theorem}\label{thm:Mod:comp}
	Let 
	$\mod_1$ and $\mod_2$ be \mcomps
	and
	$X \subseteq \left(\inpof{\mod_1}\cap \outpof{\mod_2}\right) $ non-empty.
	Then 
	$$
	\begin{array}{c}
		\mod_2 \mbox{ expands }\mod_1
		\\
		\mbox{on $X$}
	\end{array}
	\iff
	\left\{\mbox{\begin{tabular}{l@{\;}r@{\;}l}
		$X\neq \bordof{\mod_1}\cup \bordof{\mod_2}$
		\\
		$\restr{\behof{\mod_1}}  X \wperp \restr{\behof{\mod_2}}  X $
		\\
		each $x\in X$ is connected
		&
		either& 
		to a $y\in \left(\bordof{\mod_1}\setminus X\right)$ in each test of $\mod_1$,
		\\&
		or &
		to a $z\in \left(\bordof{\mod_2}\setminus X\right)$ in each test of $\mod_2$
	\end{tabular}
	}\right.
	$$
\end{theorem}
\begin{proof}
	By definition of \mcomp, letting 
	$\mod=\mod_1\comp X \mod_2$.
	\begin{itemize}
		
		\item[$\left(\Rightarrow\right)$]
		If $\mod$ is a \mcomp, then $\bordof\mod = \left(\bordof{\mod_1}\cup\bordof{\mod_2}\right)\setminus X$ must be non-empty, 
		thus $X\neq\left(\bordof{\mod_1}\cup\bordof{\mod_2}\right)$.

		If we assume that $\restr{\behof{\mod_1}}{X}\not \wperp \restr{\behof{\mod_2}}{X}$, 
		then there are $x,y\in X$ such that $\set{x,y}\subset\gamma_1 \in p \in \restr{\behof{\mod_1}}X$ and $\set{x,y}\subset\gamma_2 \in q \in \restr{\behof{\mod_2}}X$. 
		That is, there is a path between $x$ and $y$ in both $\switch_1(\mod_1)$ and $\switch_2(\mod_2)$ with $\switch_i\in\switchset{\mod_i}$ and $i\in\set{1,2}$. 
	Thus there is a switch $\switch\in \switchset\mod$ such that $\switch(\mod)$ contains a cycle.
		This contradicts the hypothesis of $\mod$ being a \mllmod.
		
		Finally, since each vertex of $\mod$ is connected to a vertex in $\bordof\mod=\left(\bordof{\mod_1}\cup\bordof{\mod_2}\right)\setminus X$ in each test,
		then, each $x\in X$ is.

		\item[$\left(\Leftarrow\right)$]
		Since 
		$\mod_1$ and $\mod_2$ are \mcomps, 
		then $\switch(\mod_1)$ and $\switch(\mod_2)$ does not contain cycles (where, for both $i\in\set{1,2}$.
		Let us denote by $\sigma(\mod_i)$ the test defined from $\mod_i$ by the switch obtained by restringing $\sigma\in\switchset{\mod}$ to $\mod_i$).
		Then, if $\switch(\mod)$ contains a cycle, then 
		there are $x,y\in X$ such that $x$ and $y$ are connected in $\sigma(\mod_1)$ and in $\sigma(\mod_2)$.
		This implies the existence of a $\gamma_1\in p \in \restr{\behof{\mod_1}}X$ and a $\gamma_2\in q \in \restr{\behof{\mod_2}}X$ with $\gamma_1$ and $\gamma_2$ both containing $x$ and $y$, therefore $p\not \wperp q$, contradicting the hypothesis.
		
		The fact that each vertex in $\mod$ is connected to a vertex of the border in any test is consequence of the fact that each vertex in $\mod$ (then, in particular, each $x\in X$) is connected to a vertex in 
		$\bordof\mod=\left(\bordof{\mod_1}\setminus X\right) \cup \left(\bordof{\mod_2}\setminus X\right)$.%
		\qedhere
	\end{itemize}
\end{proof}

\begin{corollary}\label{cor:trans:comp}
	Let $\mod$ be a \mcomp, and $\tmod$ be \tracomp
	and
	$\left(\inpof{\mod}\cap \outpof{\tmod}\right)\supseteq X \neq \emptyset$.
	If $\mod$ is a transitory, then $\tmod\comp X \mod$ is so.
\end{corollary}
\begin{proof}
	In light of \Cref{thm:Mod:comp}, 
	it suffices to remark that if 
	$\mod$ is a \tracomp, 
	then every input $\ip\in\inpof{\mod}$ is connected to an output $\op\in\outpof\mod$ in a test $\switch_2(\mod)$.
	If $\op\in\outpof\mod$ we conclude.
	Otherwise $\op\notin\outpof\mod$ and $\op=x\in X$.
	Since $\tmod$ is \tracomp, then by definition there is a test $\switch_1(\tmod)$ in which 
	$x$ is connected to a vertex $\op\in\outpof\tmod$.
	We conclude since 
	each input of $\tmod\comp X \mod$ is either an input of $\tmod$, or an input of $\mod$;
	and in the latter case we have a switching $\switch$ 
	defined as the union of $\switch_1$ and $\switch_2$ such that $\ip$ is connected to an output $\op\in\outpof\tmod\subseteq\outpof{\tmod\comp X \mod}$.
\end{proof}

\begin{proposition}\label{thm:complexity}
	We can check if a \mcomp $\mod'$ expands a \mcomp $\mod$
	in polynomial time with respect to 
	$\sizeof{\outpof{\mod'}} + \sizeof{\behof{\mod'}} + \sizeof{\behof{\mod}}$.
\end{proposition}
\begin{proof}
	To check if $\behof{\mod'}\wperp \restr{\behof{\mod}}{\outpof \link}$ requires $\sizeof{\behof{\mod'}} \times \sizeof{\restrp{\behof\mod}{\outpof{\mod'}}}\leq \sizeof{\behof{\mod'}} \times \sizeof{\behof{\mod}}$ orthogonality tests on partitions.
	Each test requires to build the graph $\gG(p,q)$ (linear on $\sizeof{\outpof\link}$) and  check graph acyclicity since graph traversal is linear in $\sizeof{\vertices[\gG(p,q)]}+\sizeof{\uedge[\gG(p,q)]}\leq 2\sizeof{\outpof\link}$ and $\sizeof{\uedge[\gG(p,q)]}=\sizeof{\outpof\link}$, see \cite{cormen01introduction}.
\end{proof}

\section{Modelling behaviors beyond the scope of $\MLL$-\mllpss}\label{sec:examples}

In this section we recall the definition of two classes of generalized multiplicative connectives from \cite{acc:mai:20}, showing how the corresponding links can be used to define methods 
whose behaviors exhibit unexpected context-sensitive and non-linear characteristics in a multiplicative setting.

\begin{definition}\label{def:gircon}
	Let $n=uv$ be the product of two prime numbers $u,v\in \N$.
	We define a \defn{basic partition with $u$ blocks of $v$ elements} to be a partition $p\in\partofn$ such that each block $\gamma \in p$  is either of the form
	$\gamma
	=
	\block{i, i+1, \dots, i+v-1}
	$
	if
	$i+v< n$, 
	or of the form
	$
	\gamma
	=
	\block{i, \dots, n, 1, 2, \dots, i+v-(n+1)}
	$
	for a $i \in\intset1n$ otherwise.
	We denote by $\sbp u v\subset\partofn$ the set of basic partitions with $u$ blocks of $v$ elements and $\psbp uv$ its orthogonal set of partitions.
	
	For every $n=uv$ we define the following sets of partitions of the set $\intset0n$:
	$$
	\begin{array}{l@{\;=\;}l}
		\gsbp uv(0,1,\dots, n)
		&
		\bigcup_{k=1}^{u}\Set{p_k=\part{\gamma_{p_k}, \block{i_1},\dots, \block{i_{n-u}}} \mid \restr{p_k}{\set{1,\dots, n}}\in\psbp uv \mbox{ and } 0\in \gamma_{p_k}  }
		\\
		\gsbpp uv(0,1,\dots, n)
		&
		\bigcup_{k=1}^{u}\bigcup_{j=1}^{v}
		\Set{p_k^j=\part{\gamma_1, \dots, \gamma_v} \mid \restr{p_k^j}{\set{1,\dots, n}}\in\sbp uv \mbox{ and } 0\in \gamma_j  }
	\end{array}
	$$
	and we define the following \defn{Girard's types}: 
	$\gircon_{u,v}\coloneqq	\tuple{uv, 1, \gsbp uv}	$
	and
	$\cneg\gircon_{u,v}	\coloneqq	\tuple{uv, 1, \gsbpp uv}$.
	\footnote{In \cite{acc:mai:20} it has been shown that there is no \mllmod $\mod$ containing only $\lpar$- and $\ltens$-link such that $\behof{\mod}=\behof{\gircon_{u,v}}$ or $\behof{\mod}=\behof{\cneg\gircon_{u,v}}$.}
\end{definition}

\begin{remark}\label{rem:cyclic}
	
	The behavior 
	$\gsbp uv$ is the same of the intersection of the behavior of specific \DNF-formulas,
	while 
	$\gsbpp uv$ is the same of the union of behavior of specific \CNF-formulas
	(see \Cref{fig:basic}).
	More precisely,
	$$
	\gsbp uv= \bigcap_{\tau \in \cyofn} \swof{\dnf(i_{\tau(1)}, \dots,i_{\tau(n)})}
	\qquad \mbox{and}\qquad
	\gsbpp uv= \bigcup_{\tau \in \cyofn} \swof{\cnf(i_{\tau(1)}, \dots,i_{\tau(n)})}
	\qquad
	\mbox{where}
	$$
	\begin{itemize}
		\item 
		$\swof{\dnf(1, \dots,n)}$ be the behavior of the \ms representing the formula tree of the disjunctive normal form 
		formula $\dnf(1, \dots,n)=\big(a_1\ltens \cdots \ltens a_v\big) \lpar \cdots \lpar \big(a_{n-v+1} \ltens \cdots \ltens a_n\big)$ ;
		
		\item 
		$\swof{\cnf(1, \dots,n)}$ be the behavior of the \ms representing the formula tree of the 
		conjunctive normal form
		formula $\cnf(a_1, \dots,a_n)=\big(a_1\lpar \cdots \lpar a_v\big) \ltens \cdots  \ltens \big(a_{n-v+1} \lpar \cdots \lpar a_n\big)$ ;
		
		\item 
		$\cyofn$ be the set of cyclic permutations over the set $\intset1n$ (assuming the standard order on $\N$).
	\end{itemize}
	
\end{remark}

\begin{figure}[t]
	\resizebox{\textwidth}{!}{
		\begin{tabular}{|c|c|}
			\hline
			$\gircon_{2,2}^\perp$ & $\gircon_{2,2}$\\
			\hline
			$\begin{array}{c@{}c@{}c@{}c@{}c@{}c@{}c}
				\vi1_1 &&\vi2_2 & & \vi3_3 &  & \vi4_4
				\\[.5em]
				&\Gtens1&           &        &                   &\Gtens2
				\\[.5em]
				&       &           &\Gpar1&                   &
				\\
				&       &           &\vo0
			\end{array}    
			\pswires{i1/Gtens1.L,i2/Gtens1.R,i3/Gtens2.L,i4/Gtens2.R}
			\pswires{Gtens1.O/Gpar1.L,Gtens2.O/Gpar1.R,Gpar1.O/o0}
			$
			\quad
			$
			\begin{array}{c@{}c@{}c@{}c@{}c@{}c@{}c}
				\vi1_4 &&\vi2_1 & & \vi3_2 &  & \vi4_3
				\\[.5em]
				&\Gtens1&           &        &                   &\Gtens2
				\\[.5em]
				&       &           &\Gpar1&                   &
				\\
				&       &           &\vo0
			\end{array}    
			\pswires{i1/Gtens1.L,i2/Gtens1.R,i3/Gtens2.L,i4/Gtens2.R}
			\pswires{Gtens1.O/Gpar1.L,Gtens2.O/Gpar1.R,Gpar1.O/o0}
			$
			&
			$\begin{array}{c@{}c@{}c@{}c@{}c@{}c@{}c}
				\vi1_1 &&\vi2_2 & & \vi3_3 &  & \vi4_4
				\\[.5em]
				&\Gpar1&           &        &                   &\Gpar2
				\\[.5em]
				&       &           &\Gtens1&                   &
				\\
				&       &           &\vo0
			\end{array}    
			\pswires{i1/Gpar1.L,i2/Gpar1.R,i3/Gpar2.L,i4/Gpar2.R}
			\pswires{Gpar1.O/Gtens1.L,Gpar2.O/Gtens1.R,Gtens1.O/o0}
			$
			\qquad\qquad
			$\begin{array}{c@{}c@{}c@{}c@{}c@{}c@{}c}
				\vi1_4 &&\vi2_1 & & \vi3_2 &  & \vi4_3
				\\[.5em]
				&\Gpar1&           &        &                   &\Gpar2
				\\[.5em]
				&       &           &\Gtens1&                   &
				\\
				&       &           &\vo0
			\end{array}    
			\pswires{i1/Gpar1.L,i2/Gpar1.R,i3/Gpar2.L,i4/Gpar2.R}
			\pswires{Gpar1.O/Gtens1.L,Gpar2.O/Gtens1.R,Gtens1.O/o0}
			$
			\\
			\hline
			$ \Set{\part{\block{\op,\ip_1,\ip_2}, \block{\ip_3,\ip_4}}, 
				\part{\block{\op,\ip_3,\ip_4}, \block{\ip_1,\ip_2}}}$ 
			&
			$ \Set{\part{\block{\op,\ip_1,\ip_3}, \block{\ip_2},\block{\ip_4}}, 
				\part{\block{\op,\ip_2,\ip_4}, \block{\ip_1},\block{\ip_3}},
				\part{\block{\op,\ip_1,\ip_4}, \block{\ip_2},\block{\ip_3}}, 
				\part{\block{\op,\ip_2,\ip_3}, \block{\ip_1},\block{\ip_4}}
			}$ 
			\\
			$\cup $ 
			&
			$\cap$
			\\
			$\Set{\part{\block{\op,\ip_1,\ip_4}, \block{\ip_2,\ip_3}}, 
				\part{\block{\op,\ip_2,\ip_3}, \block{\ip_1,\ip_4}}}$
			&
			$ \Set{\part{\block{\op,\ip_1,\ip_3}, \block{\ip_2},\block{\ip_4}}, 
				\part{\block{\op,\ip_2,\ip_4}, \block{\ip_1},\block{\ip_3}},
				\part{\block{\op,\ip_1,\ip_2}, \block{\ip_3},\block{\ip_4}}, 
				\part{\block{\op,\ip_3,\ip_4}, \block{\ip_1},\block{\ip_2}}
			}$ 
			\\
			&=
			\\
			&
			$ \Set{
				\part{\block{\op,\ip_1,\ip_3}, \block{\ip_2},\block{\ip_4}}, 
				\part{\block{\op,\ip_2,\ip_4}, \block{\ip_1},\block{\ip_3}}
			}$
			\\\hline
		\end{tabular}
	}
	\caption{Connective  $\gircon_{2,2}^\perp$ and its dual connective $\gircon_{2,2}$ interpreted 
		as the union of the behaviors of {\DNF} formula-trees
		and 
		as the intersection of the behaviors of {\CNF} formula-trees
		respectively.
	}
	\label{fig:basic}
\end{figure}

\begin{theorem}[\cite{acc:mai:20}]\label{thm:nonDec}
	There is no \ms $\gS$ over the signature $\set{\lpar,\ltens}$ such that $\behof{\gS}=\behof{\gircon_{u,v}}$ or $\behof{\gS}=\behof{\cneg\gircon_{u,v}}$
	for any $u,v\in\N$  prime numbers.
\end{theorem}

\begin{definition}
	We generalize the components of bipoles (see \Cref{def:concBip}) as follows:
	
	\begin{itemize}
		\item 
		A \defn{generalized body} is 
		a \mcomp made of 
		a $\lmtens$ collecting the outputs of a \ms representing a \CNF-formula (i.e., bodies) or $\gsbpp uv$-links.
		
		\item 
		A \defn{generalized synchronizer} is 
		a \mcomp made of 
		a $\lmparp$ collecting the outputs of a \ms representing a \DNF-formula (i.e., synchronizers) or $\gsbp u2$-links.
		
	\end{itemize}
	
	A \defn{generalized bipole} is 
	a \mcomp made of 
	a generalized synchronizer $\rmod$ collecting the outputs of 
	certain generalized bodies $\cmod_1,\ldots,\cmod_n$  and headers $\hmod_1, \dots, \hmod_n$
	whose structure is similar to the one in \Cref{eq:bipLink}, but where no output $F$ occurs (thanks to the $\lmparp$ in the generalized synchronizer).
\end{definition}
\begin{remark}
	Headers and Generalized bodies and synchronizes are \mmods.
\end{remark}

%

The following result is a corollary of \Cref{thm:nonDec}.
\begin{corollary}
	There are generalized bipoles whose behavior is different from any behavior of a $\MLL$-\mllps.
\end{corollary}
\begin{remark}\label{rem:outp}
	In \cite{and:maz:03} the bipolar links have an additional output (the one labeled by the formula) 
	as shown in \Cref{eq:bipolarLink}
	because the syntax from \cite{and:maz:03} 
	is closer to the representation of methods we use in \Cref{eq:bipoleSynth}, where 
	the name of the applied method (the formula $F$ identifying the rule) occurs in the final sequent.
	In our syntax this superfluous information is discarded in the same way as in \Cref{eq:bipole}, where the formula $F$ does not occur in the conclusion.
	That is, the additional output $F$ in \Cref{eq:bipLink} would not occur in our generalized bipoles because in the definition of generalized synchronizes we use $\lmparp$ instead of $\lmpar$, allowing us to formally discard this output.
\end{remark}
%
%

We conclude this section by providing two toy-examples 
describing the way a server manages access requests 
to critical sections.
The na\"ive idea behind these models is that 
if the vertex corresponding to a client is connected to the vertex corresponding to a resource, then
there is a configuration of the model in which only that specific client accesses the resource.

\begin{figure}[t]
	\centering
	\resizebox{.9\textwidth}{!}{
		$
		\begin{array}{cc@{\hspace{-20pt}}c@{\quad}ccc}
			\quad\vr1_1& \vr2_2&  \\[-10pt]
			\virt1\\[25pt]
			& &\vGmod1{\quad\cneg \gircon_{2,2}\quad}\\[15pt]
			&&\Gparpone1\\
			&&&&\virt2
			\\[-10pt]
			&&& \vc1_1 &\vc2_2
		\end{array}
		\pswires{r1/Gmod1.150,r2/Gmod1.120}
		\pswires{Gmod1/Gparpone1}
		\psaxiom{Gmod1.60}{c1}{2}{}
		\psaxiom{Gmod1.30}{c2}{2}{}
		\linkbox{irt1}{irt2}
		$
		\begin{tabular}{c}
			can be replaced
			\\
			by both
		\end{tabular}
		$
		\begin{array}{c@{}ccc@{\qquad}ccc}
			&\vr1_1& \vr2_2&  \\[-10pt]
			\virt1
			\\[20pt]
			&\Gtens1 && \Gtens2
			\\
			&&\Gparptwo1\\[5pt]
			&&&&\virt2
			\\[-10pt]
			&&& \vc1_1 &\vc2_2
		\end{array}
		\pswires{r1/Gtens1.L,r2/Gtens2.L,Gtens1.O/Gparptwo1.L,Gtens2.O/Gparptwo1.R}
		\psaxiom{Gtens1.R}{c1}{1}{}
		\psaxiom{Gtens2.R}{c2}{1}{}
		\linkbox{irt1}{irt2}
		$
		and 
		$
		\begin{array}{c@{}ccc@{\qquad}cc@{\quad}c}
			&\vr1_1& \vr2_2&  \\[-10pt]
			\virt1
			\\[20pt]
			&\Gtens1 && \Gtens2
			\\
			&&\Gparptwo1\\[5pt]
			&&&&\virt2
			\\[-10pt]
			&&& \vc1_1 &\vc2_2
		\end{array}
		\pswires{r1/Gtens2.L,r2/Gtens1.L,Gtens1.O/Gparptwo1.L,Gtens2.O/Gparptwo1.R}
		\psaxiom{Gtens1.R}{c1}{.8}{}
		\psaxiom{Gtens2.R}{c2}{1}{}
		\linkbox{irt1}{irt2}
		$
	}
	\caption{A \ms representing a non-deterministic application of two concurrent methods.}
	\label{fig:concurrentResources}
\end{figure}

\begin{figure}[t]
	\centering
	\resizebox{.9\textwidth}{!}{
		$
		\begin{array}{cc@{\hspace{-10pt}}c@{\hspace{-10pt}}ccc}
			\quad\vr1_1& \vr2_2 & & \vr3_3 &\vr4_4 \\[-10pt]
			\virt1\\[10pt]
			& &\vGmod1{\quad \gircon_{2,2}\quad}\\[5pt]
			&&&&\Gtensptwo1
			\\
			&&&&& \quad\virt2
			\\[-10pt]
			&&&&& \vc1
		\end{array}
		\pswires{r1/Gmod1.150,r2/Gmod1.120,r3/Gmod1.60,r4/Gmod1.30}
		\pswires{Gmod1/Gtensptwo1.L}
	\psaxiom{Gtensptwo1.R}{c1}{2}{}
	\linkbox{irt1}{irt2}
	$
	\begin{tabular}{c}
		can replace
	\end{tabular}
	$
		\begin{array}{ccc@{\hspace{-5pt}}c@{\hspace{-5pt}}cccccc}
			\quad
			\vr1_1&& \vr2_2 && \vr3_3 &&\vr4_4 \\[-10pt]
			\virt1\\[10pt]
			&\Gpar1 &&&&\Gpar2\\[5pt]
			&&&\Gtens3
			\\
			&&&&&\Gtensptwo1
			\\
			&&&&&&& \quad\virt2
			\\[-10pt]
			&&&&&&& \vc1
		\end{array}
		\psaxiom{Gtensptwo1.R}{c1}{2}{}
		\pswires{r1/Gpar1.L,r2/Gpar1.R,r3/Gpar2.L,r4/Gpar2.R}
	\pswires{Gpar1.O/Gtens3.L,Gpar2.O/Gtens3.R,Gtens3.O/Gtensptwo1.L}
	\linkbox{irt1}{irt2}
$
or 
$
	\begin{array}{ccc@{\hspace{-5pt}}c@{\hspace{-5pt}}cccccc}
		\quad
		\vr1_1&& \vr2_2 && \vr3_3 &&\vr4_4 \\[-10pt]
		\virt1\\[10pt]
		&\Gpar1 &&&&\Gpar2\\[5pt]
		&&&\Gtens3
		\\
		&&&&&\Gtensptwo1
		\\
		&&&&&&& \quad\virt2
		\\[-10pt]
		&&&&&&& \vc1
	\end{array}
	\psaxiom{Gtensptwo1.R}{c1}{2}{}
	\pswires{r1/Gpar1.L,r2/Gpar2.L,r3/Gpar2.R}
\pswire{r4}{Gpar1.R}{.5}{}
\pswires{Gpar1.O/Gtens3.L,Gpar2.O/Gtens3.R,Gtens3.O/Gtensptwo1.L}
\linkbox{irt1}{irt2}
$}
\vspace{-10pt}
\caption{A link representing a method application with a dependent choice.}
\label{fig:dependentChoice}
\end{figure}

\begin{example}[\basicul]
Consider a server receiving a request from two different clients $c_1$ and $c_2$ to access, at the same time, to one resource $r_1$ or $r_2$ in a critical section.
In this case the server can execute four different methods 
of the form 
$\prolog{r_i \colonminus c_j}$
each of which represents the resource $r_i$ being allocated to the client $c_j$ (for some $i,j\in \set{1,2}$).
Once any one of these methods is executed, the condition of critical section requires that no other user can access to this resource (until it is released).
Therefore, either the server authorizes $c_1$ to access $r_1$, and authorizes $c_2$ to access $r_2$, 
or the server authorizes $c_1$ to access $r_2$, and authorizes $c_2$ to access $r_1$.

In both cases, the two methods representing the clients accessing the resources can be applied concurrently, and the \mss on the right-hand side of \Cref{fig:concurrentResources} represent these two configurations.
Note that none of the two \mss fully capture the described configuration: each solution makes a choice about which client has access to which resource since, and this kind of choices are beyond the scope of multiplicative linear logic.
Using the \basicul ${\cneg \gircon_{2,2}}$ we can define the \ms in the left-hand side of \Cref{fig:concurrentResources},
whose behavior is the same of the union of the two \mss describing the two possible choices (see \Cref{rem:cyclic}).
That is, this the ${\cneg \gircon_{2,2}}$-link can be interpreted as a synchronizer allowing such a concurrent choice.
\end{example}

\begin{example}[\basicil]
Consider a server receiving a request from a single client $c$ to access the set of resources $\set{r_1,r_2,r_3,r_4}$
with goal of either collect $r_1$ and $r_3$ (that is, to apply the method $\prolog{c\colonminus r_1, r_3}$), or to collect $r_2$ and $r_4$ (that is, to apply the method $\prolog{c\colonminus r_2, r_4}$).
Because of our constraints, the server can
either
grant access to a resource in $\set{r_1,r_2}$ and one in $\set{r_3,r_4}$,
or 
grant access to a resource in $\set{r_1,r_4}$ and one in $\set{r_2,r_3}$.
These two different solutions are represented by the \mss in the right-hand side of \Cref{fig:dependentChoice}.
However, if we consider singularly each \ms, it models more permissive configurations. 
For example, the right-most \ms admits a test in which $c$ is connected to $r_1$ and $r_4$, which is not a configuration we want to allow.
Indeed, the configurations we want to allow are exactly the configurations which are valid in both \mss.
However, using 
the \basicil ${\gircon_{2,2}}$ we can define the \ms in the left of \Cref{fig:dependentChoice},
whose behavior is the same of the intersection of the two \mss on the right-hand side:
the client $c$ can only access at the same time to either $r_1$ and $r_3$, or $r_2$ and $r_4$.
\end{example}

\section{Conclusion}\label{sec:conc}

In this paper, we extended the multiplicative fragment of the logic programming framework studied in \cite{and:01,and:02,and:maz:03,fou:mog04,hae:fag:sol:07}. 
Within this framework, as  main novelty, we offered a computational interpretation of the generalized connectives discussed in \cite{mai:17,acc:mai:20}. 
These connectives were initially introduced in early works on linear logic, but prior to this work, they had not been given a concrete computational interpretation:
they cannot be expressed using combinations of the multiplicative connectives $\lpar$ and $\ltens$
and 
they describe ``locally additive'' behaviors \cite{lau-mai:08,abr-mai:2016} such as non-deterministic or conditional choices.
It is worth noting that the methodology used to define our framework appears to align with the definition of the basic building block of the \emph{transcendental syntax} \cite{gir:trans,eng:sei:trans} and its extensions \cite{eng:sei:stellar:Ar,eng:sei:stellar,eng:sei:tile}.

\myparagraph{Future works}
%
As observed in \Cref{rem:cyclic}, the behavior of a \basicul can be seen as the union of behaviors of bodies (see \Cref{fig:basic}). This allows us to replace any \basicul within the \multstr of any of these bodies, and preserving correctness. 
Leveraging this intuition, we could define a non-deterministic \emph{\unfold} operation on \basicul, returning their set of bodies. 
This operation can be further extended to \mss, providing a notion of expansion for multiplicative structures similar to the concept of \emph{Taylor expansion} in \emph{differential linear logic}~\cite{ehr:reg:dif}. 
Such an expansion, could be interpreted as representing specific instances of \emph{delayed choice}~\cite{bae:mau:delayed}: during proof (net) construction we do not need to specify which of the possible bodies we want to use in a specific step, but we can use a \basicul containing it instead, delaying this decision.
We also envision an extension of our model where links have probabilistic distributions on the set of switchings. In this setting \basiculs could be equipped with probability distribution functions, transforming the non-deterministic expansion operator into a probabilistic one. Consequently, multiplicative structures could be employed to model Bayesian networks \cite{pearl:book,maieli:lirmm-03271511}.
Additionally, we foresee the possibility of defining refinements of the \emph{attack trees} syntax with a linear treatment of resources but including specific non-deterministic choices \cite{mau:oos:attack,hor:mau:tiu:attack,trees}. 
Similarly, the linear constraints on the hypergraphs used in our model allow us to define a concurrent computational model with a more granular management of resource consumption, akin to what we experience in the management of critical sections in concurrent systems.
Another possible direction is to study modules for proof structures built using the graphical connectives from \cite{acc:hor:str:LICS20,acc:hor:str:LMCS,acc:hor:mau:str:time,acc:GPT1} to provide them with a computational meaning based on resources separation.


\bibliographystyle{eptcs}
\bibliography{refMmod}

\end{document}